\documentclass[11pt]{article}

\usepackage[overload]{empheq}
\usepackage{latexsym,dsfont,textcomp,amsmath}
\usepackage[english]{babel}
\usepackage[latin1]{inputenc}
\usepackage{enumerate,indentfirst}
\usepackage[colorlinks,citecolor=blue,urlcolor=blue]{hyperref}
\usepackage{amssymb}
\usepackage{amsthm}
\usepackage{hyperref}

\newtheorem{Theorem}{Theorem}[part]
\newtheorem{Definition}{Definition}[part]
\newtheorem{Proposition}{Proposition}[part]
\newtheorem{Assumption}{Assumption}[part]
\newtheorem{Lemma}{Lemma}[part]
\newtheorem{Corollary}{Corollary}[part]
\newtheorem{Remark}{Remark}[part]

\newcommand{\nc}{\newcommand}
\nc{\esssup}{\mathop{\mathrm{ess\,sup}}}
\nc{\essinf}{\mathop{\mathrm{ess\,inf}}}
\nc{\argmax}{\mathop{\mathrm{arg\,max}}}

\def \P{\mathbb{P}}
\def \N{\mathbb{N}}
\def \R{\mathbb{R}}

\def \E{\mathbb{E}}

\def \G{\mathbb{G}}
\def \Q{\mathbb{Q}}

\def \1{\mathds{1}}

\def \Ac{{\cal A}}
\def \Bc{{\cal B}}
\def \Cc{{\cal C}}
\def \Dc{{\cal D}}

\def \Fc{{\cal F}}
\def \Gc{{\cal G}}

\def \Sc{{\cal S}}

\begin{document}

\title{Exponential utility maximization and  indifference price in an incomplete market with defaults}
\author{Thomas LIM \thanks{Laboratoire de Probabilit\'es et Mod\`eles Al\'eatoires,
CNRS, UMR 7599, Universit{\'e}s Paris 6-Paris 7, \sf
tlim@math.jussieu.fr}
~~~ Marie-Claire QUENEZ\thanks{Laboratoire de Probabilit\'es et Mod\`eles Al\'eatoires,
CNRS, UMR 7599, Universit{\'e}s Paris 6-Paris 7, \sf
quenez@math.jussieu.fr}
}

\maketitle
\begin{abstract}
In this paper, we study the indifference pricing of a contingent claim via the maximization of exponential utility over a set of admissible strategies. We consider a financial market with a default time inducing a discontinuity in the price of stocks. We first consider the case of strategies valued in a {\em compact} set. Using a verification theorem, we show that in the case of bounded coefficients the value function of the exponential utility maximization problem can be characterized as {\em the solution of a Lipschitz BSDE} (backward stochastic differential equation). Then, we consider the case of non constrained strategies. By using dynamic programming technics, we state that the value function is the {\em maximal subsolution of a BSDE}. Moreover, the value function is the {\em limit of a sequence of processes}, which are the value functions associated with some subsets of bounded admissible strategies. In the case of bounded coefficients, these approximating processes are the solutions of Lipschitz BSDEs, which leads to possible numerical computations. These properties can be applied to the indifference pricing problem and they can be generalized to the case of several default times or a Poisson process.
\end{abstract}

\vspace{1cm}



\textbf{MSC Classification (2000):} 49L20, 93E20.
\newpage

\section{Introduction}
\setcounter{equation}{0} \setcounter{Assumption}{0}
\setcounter{Theorem}{0} \setcounter{Proposition}{0}
\setcounter{Corollary}{0} \setcounter{Lemma}{0}
\setcounter{Definition}{0} \setcounter{Remark}{0}

In this paper, we study the indifference pricing problem in a market where the underlying traded assets are assumed to be local martingales driven by a Brownian motion and a default indicating process. We denote by $S_t=(S_t^i)_{1\leq i \leq n}$ for all $t\in [0,T]$ the price of these assets where $T<\infty$ is the fixed time horizon and $n$ is the number of assets. The price process $(S_t)$ is defined on a filtered space $(\Omega,\Gc,(\Gc_t)_{0\leq t \leq T},\P)$. Following Hodges and Neuberger \cite{hodneu89}, we define the (buying) indifference price $p(\xi)$ of a contingent claim $\xi$, where $\xi$ is a $\Gc_T$-measurable random variable, as the implicit solution of the equation
\begin{equation}
\sup_{\pi}\E\Big[U\Big(x+\int_0^T\pi_tdS_t\Big)\Big]=\sup_{\pi}\E\Big[U\Big(x-p(\xi)+\int_0^T\pi_tdS_t+\xi\Big)\Big],
\label{equation indifference}
\end{equation}
where the suprema are taken over admissible portfolio strategies $\pi$, $x\in \R$ is the initial endowment and $U$ is a given utility function. In other words, the price of the contingent claim is defined as the amount of money $p(\xi)$ to withdraw to his initial wealth $x$ that allows the investor to achieve the same supremum of the expected utility as the one he would have had with initial wealth $x$ without buying the claim. 
A lot of papers study the indifference pricing problem. Among them, we quote Rouge and El Karoui \cite{rounek00} for a Brownian filtration, Biagni \emph{et al.} \cite{biafrigra08} for the case of general semimartingales, Bielecki and Jeanblanc \cite{biejea08} for the case of a discontinuous filtration. An extensive survey of the recent literature on this topic can be found in Carmona \cite{car08}. \\

Throughout this paper, the utility function $U$ is assumed to be the exponential utility. By (\ref{equation indifference}), the study of the indifference pricing of a given contingent claim is clearly linked to the study of the utility maximization problem. 

Recall that concerning the study of the maximization of the utility of terminal wealth, there are two possible approaches:
\begin{itemize}
\item The first one is the {\em dual approach} formulated in a {\em static} way.
This dual approach has been largely studied in the literature. Among them, in a Brownian framework, we quote Karatzas \emph{et al.} \cite{karlehshr87} in a complete market and Karatzas \emph{et al.} \cite{karlehshrxue91} in an incomplete market. In the case of general semimartingales, we quote Kramkov and Schachermayer \cite{krasch99}, Shachermayer \cite{sch01} and Delbaen \emph{et al.} \cite{deletal02} for the particular case of an exponential utility function. For the case with a default in a markovian setting we refer to Lukas \cite{luk01}. Using this approach, these different authors solve the utility maximization problem in the sense of finding the optimal strategy and also give a characterization of the optimal strategy via the solution of the dual problem.
\item The {\em second approach} is the {\em direct} study of the primal problem(s) by using stochastic control technics such as {\em dynamic programming}. 
Recall that these technics had been used in finance but only in a markovian setting for along time. For example the reference paper of Merton \cite{mer71} uses the well known Hamilton-Jacobi-Bellman verification theorem to solve the utility maximization problem of consumption/wealth in a complete market. The use in finance of stochastic dynamic technics (presented in El Karoui's course \cite{nek79} in a general setting) is more recent. One of the first work in finance using these technics is that of El Karoui and Quenez \cite{nekque95}. Also, recall that the backward stochastic differential equations (BSDEs) have been introduced by Duffie and Epstein \cite{dufeps92} in the case of recursive utilities and by Peng \cite{pen90} for a general Lipschitz coefficient. In the paper of El Karoui \emph{et al.} \cite{nekpenque97}, several applications to finance are presented. Also, an interesting result of this paper is a {\em verification} theorem which allows to characterize the dynamic value function of an optimization problem as the solution of a Lipschitz BSDE. This principle stated in the Brownian case has many applications in finance. One of them can be found in Rouge and El Karoui \cite{rounek00} who study the exponential utility maximization problem in the incomplete Brownian case and characterize the dynamic indifference price as the solution of a quadratic BSDE (introduced by Kobylanski \cite{kob00}). Concerning the exponential utility maximization problem, there is also the nice work of Hu \emph{et al.} \cite{huimkmul05} still in the Brownian case. By using a {\em verification} theorem (different from the previous one), they characterize the logarithm of the dynamic value function as the solution of a quadratic BSDE. 
\end{itemize}

The case of a discontinuous framework is more difficult. One reason is that there are less results on BSDEs with jumps than in the Brownian case. Concerning the study of the exponential utility maximization problem in this case, we refer to Morlais \cite{mor09}. She supposes that the price process of stock is modeled by a local martingale driven by an independent Brownian motion and a Poisson point process. She mainly studies the interesting case of admissible strategies valued in a compact set (not necessarily convex). Using the same approach as in Hu \emph{et al.} \cite{huimkmul05}, she states that the logarithm of the associated value function is the unique solution of a quadratic BSDE (for which she shows an existence and a uniqueness result). In the non constrained case, she obtains formally a quadratic BSDE. She proves the existence of a solution of this BSDE by using an approximation method but she does not obtain uniqueness result. Hence, in this case, this does not allow to characterize the value function in terms of BSDEs.\\
\indent In this paper, we first consider the case of strategies valued in a {\em compact} set. By using a 
verification theorem, which is a generalization of that of El Karoui \emph{et al.} \cite{nekpenque97} to the case of jumps, we show that the value function of the exponential utility maximization problem can be characterized as the solution of a {\em Lipschitz BSDE}. Second, we consider the case of {\em non constrained} strategies. We use the dynamic programming principle to show directly that the value function is characterized as the maximal solution or the {\em maximal subsolution} of a BSDE. Moreover, we give {\em another characterization} of the value function as the {\em nonincreasing limit of a sequence of processes}, which are the value functions associated with some subsets of bounded admissible strategies. In the case of bounded coefficients, these approximating processes are the solutions of Lipschitz BSDEs. As a direct consequence, this suggests some possible {\em numerical computations} in order to approximate the value function and the indifference price. Also, we generalize these results to the case of several default times and several stocks, and to the case of a Poisson process instead of a hazard process.

 The outline of this paper is organized as follows. In Section 2, we present the market model and the maximization problem in the case of only one risky asset ($n=1$). In Section 3, we study the case of strategies valued in a {\em compact} set. In Section 4, we consider the {\em non constrained case} and  state a first characterization of the value function as the maximal subsolution of a BSDE. In Section 5, we give a second characterization of the value function as the nonincreasing limit of a sequence of processes. In Section 6, we consider the classical case where the {\em coefficients are bounded} which simplifies the two previous characterizations of the value function. In Section 7, we study the case of unbounded {\em coefficients} which satisfy {\em some exponential integrability} conditions. Finally in Section 8, we study the {\em indifference price} for a contingent claim. In the last section, we generalize the previous results to the case of several assets ($n\geq 1$) and several default times and we also extend these results to a Poisson jump model.

\section{The market model}
\label{modele}
\setcounter{equation}{0} \setcounter{Assumption}{0}
\setcounter{Theorem}{0} \setcounter{Proposition}{0}
\setcounter{Corollary}{0} \setcounter{Lemma}{0}
\setcounter{Definition}{0} \setcounter{Remark}{0}

Let $(\Omega, \mathcal{G},\mathbb{P})$ be a complete probability space. We assume that all processes are defined on a finite time horizon $[0,T]$. Suppose that this space is equipped with two stochastic processes: a unidimensional standard Brownian motion $(W_t)$ and a jump process $(N_t)$ defined by $N_t=\1_{\tau\leq t}$ for any $t\in[0,T]$, where $\tau$ is a random variable which modelizes a default time (see Section \ref{plusieur defaut} for several default times). We assume that this default can appear at any time that is $\P(\tau>t)>0$ for any $t\in[0,T]$. We denote by $\mathbb{G}=\{\mathcal{G}_t,0\leq t \leq T\}$ the completed filtration generated by these processes. The filtration is supposed to be right-continuous and $(W_t)$ is a $\mathbb{G}$-Brownian motion.\\
\indent We denote by $(M_t)$ the compensated martingale of the process $(N_t)$ and by $(\Lambda_t)$ its compensator. We assume that the compensator $(\Lambda_t)$ is absolutely continuous with respect to Lebesgue's measure, so that there exists a process $(\lambda_t)$ such that $\Lambda_t=\int_0^t\lambda_sds$. Hence, the $\mathbb{G}$-martingale $(M_t)$ satisfies
\begin{equation}
M_t=N_t-\int_0^t\lambda_sds\,.
\label{M}
\end{equation}

We introduce the following sets:
\begin{itemize}
\item $\Sc^{+,\infty}$ is the set of positive $\G$-adapted $\P$-essentially bounded c\`ad-l\`ag processes on $[0,T]$.
\item $L^{1,+}$ is the set of positive $\G$-adapted c\`ad-l\`ag processes on $[0,T]$ such that $\E[Y_t]<\infty$ for any $t\in [0,T]$.
\item $L^2(W)$ (resp. $L^2_{loc}(W)$) is the set of $\G$-predictable processes on $[0,T]$ under $\P$ with
 \begin{equation*}
 \mathbb{E}\Big[\int_0^T|Z_t|^2dt\Big]<\infty~~\text{(resp. }\int_0^T|Z_t|^2dt<\infty~a.s. \text{ )}.
 \end{equation*}
\item $L^2(M)$ (resp. $L^2_{loc}(M)$, $L^1_{loc}(M)$) is the set of $\G$-predictable processes on $[0,T]$ such that 
\begin{equation*}
\E\Big[\int_0^T\lambda_t|U_t|^2dt\Big]<\infty~~\text{(resp. }\int_0^T\lambda_t|U_t|^2dt<\infty, \int_0^T\lambda_t|U_t|dt<\infty~a.s. \text{ ).}
\end{equation*}\end{itemize}

We recall the useful martingale representation theorem (see Jeanblanc \emph{et al.} \cite{jeayorche09}):
\begin{Lemma}\label{theoreme representation}
Any $(\P,\G)$-local martingale has the representation 
\begin{equation}\label{equation representation}
m_t=m_0+\int_0^ta_sdW_s+\int_0^tb_sdM_s, ~\forall\,t\in[0,T] ~ a.s.,
\end{equation}
where $a\in L^2_{loc}(W)$ and $b\in L^1_{loc}(M)$. If $(m_t)$ is a square integrable martingale, each term on the right-hand side of the representation (\ref{equation representation}) is square integrable. 
\end{Lemma}

 We now consider a financial market which consists of one risk-free asset, whose price process is assumed for simplicity to be equal to $1$ at any date, and one risky asset with price process $S$ which admits a discontinuity at time $\tau$ (we give the results for $n$ assets and $p$ default times in Section \ref{plusieur defaut}). In the sequel, we consider that the price process $S$ evolves according to the equation
\begin{equation}
dS_t=S_{t^-}(\mu_tdt+\sigma_tdW_t+\beta_tdN_t),
\label{actif S}
\end{equation}
with the classical assumptions:
\begin{Assumption}\hfill

{\rm \begin{enumerate}[(i)]
\item $(\mu_t)$, $(\sigma_t)$ and $(\beta_t)$ are $\G$-predictable processes such that $\sigma_t > 0$ and
\begin{equation*}
\int_0^T|\sigma_t|^2dt+\int_0^T\lambda_t|\beta_t|^2dt < \infty~a.s.,
\end{equation*}
\item the process $(\beta_t)$ satisfies $\beta_\tau>-1$ (this assumption implies that the process $S$ is positive).
\end{enumerate}
}\end{Assumption}

We also suppose that $\E[\exp(-\int_0^T\alpha_s dW_s -\frac{1}{2}\int_0^T \alpha^2_t dt)] =1$ where $\alpha_t = (\mu_t+\lambda_t\beta_t)/\sigma_t$, which gives the existence of a martingale probability measure and hence the absence of arbitrage.

A $\G$-predictable process $\pi=(\pi_t)_{0\leq t\leq T}$ is called a trading strategy if $\int_0^T\frac{\pi_t}{S_{t^-}}dS_t$ is well defined, e.g. $\int_0^T|\pi_t\sigma_t|^2dt +\int_0^T\lambda_t|\pi_t\beta_t|^2dt<\infty$ a.s. The process $(\pi_t)_{0\leq t\leq T}$ describes the amount of money invested in the risky asset at time $t$. The wealth process $(X^{x,\pi}_t)$ associated with a trading strategy $\pi$ and an initial capital $x$, under the assumption that the trading strategy is self-financing, satisfies the equation
\begin{equation}\label{richesse}
\left\{\begin{aligned}
dX_t^{x,\pi} & =\pi_t\big(\mu_tdt+\sigma_tdW_t+\beta_tdN_t\big), \\
X_0^{x,\pi} & =x.
\end{aligned}\right.
\end{equation}
For a given initial time $t$ and an initial capital $x$, the associated wealth process is denoted by $X^{t,x,\pi}_s$.\\
\indent We assume that the investor in this financial market faces some liability, which is modeled by a random variable $\xi$ (for example, $\xi$ may be a contingent claim written on a default event, which itself affects the price of the underlying asset). We suppose that $\xi\in L^2(\Gc_T)$ and is non-negative (note that all the results still hold under the assumption that $\xi$ is only bounded from below). 

Our aim is to study the classical optimization problem
\begin{equation}\label{pb exponentiel}
V(x, \xi)=\sup\limits_{\pi \in \Dc} \mathbb{E}\big[U(X_T^{x,\pi}+\xi)\big],
\end{equation}
where $\Dc$ is a set of admissible strategies (independent of $x$) which will be specified in the sequel and $U$ is an exponential utility function
 \begin{equation*}
 U(x)=-\exp(-\gamma x), ~x \in \R,
  \end{equation*} 
 where $\gamma >0$ is a given constant, which can be seen as a coefficient of absolute risk aversion. Hence, the optimization problem (\ref{pb exponentiel}) can be clearly written as
\begin{equation*}
V(x,\xi)=e^{-\gamma x}V(0,\xi).
\end{equation*}
Hence, it is sufficient to study the case $x= 0$. To simplify notation we will denote $X^{\pi}_t$ (resp. $X^{t,\pi}_t$) instead of $X^{0,\pi}_t$ (resp. $X^{t,0,\pi}_t$). Also, note that 
\begin{equation}\label{pb exponentieldeux}
V(0, \xi)= - \inf \limits_{\pi \in \Dc} \mathbb{E}\big[\exp\big(-\gamma(X^{\pi}_T+\xi)\big)\big].
\end{equation}

\section[Strategies valued in a compact set]{Strategies valued in a given compact set (in the case of bounded coefficients)}
\label{compact}
\setcounter{equation}{0} \setcounter{Assumption}{0}
\setcounter{Theorem}{0} \setcounter{Proposition}{0}
\setcounter{Corollary}{0} \setcounter{Lemma}{0}
\setcounter{Definition}{0} \setcounter{Remark}{0}

In this section, we study the case where the strategies are constrained to take their values in a compact set denoted by $C$ (the admissible set will be denoted by $\Cc$ instead of $\Dc$). 
\begin{Definition}{\rm
The set of admissible strategies $\Cc$ is the set of predictable $\R$-valued processes $\pi$ such that they take their values in a compact set $C$ of $\R$.
}\end{Definition}  
We assume in this part that:
\begin{Assumption}{\rm
The processes $(\mu_t)$, $(\sigma_t)$, $(\beta_t)$ and the compensator $(\lambda_t)$ are uniformly bounded.
\label{coefficient borne}
}\end{Assumption}
 This case cannot be solved by using the dual approach because the set of admissible strategies is not necessarily convex. In this context, we address the problem of characterizing dynamically the value function associated with the exponential utility maximization problem.
 We give a dynamic extension of the initial problem (\ref{pb exponentieldeux}) (with $\Dc$ $=$ $\Cc$). For any initial time $t\in [0,T]$, we define the value function $J(t,\xi)$ (also denoted by $J(t)$) by the following random variable
\begin{equation}\label{Jkt}
J(t,\xi)=\essinf_{\pi\in \Cc_t}\E\big[\exp\big(-\gamma(X^{t,\pi}_T+\xi)\big)\big|\Gc_t\big],
\end{equation}
where $\Cc_t$ is the set of predictable $\R$-valued processes $\pi$ beginning at $t$ and such that they take their values in $C$. Note that $V(0, \xi) = - J(0, \xi)$. 

In the sequel, for $\xi$ fixed, we want to characterize this dynamic value function $J(t)$ $(= J(t,\xi))$ as the solution of a BSDE. 

For that, for each $\pi \in \Cc$, we introduce the c\`{a}d-l\`{a}g process $(J_t^{\pi})$ satisfying 
\begin{equation*}
J_t^{\pi}=\mathbb{E}\big[\exp\big(-\gamma (X_T^{t,\pi}+\xi)\big)\big| \mathcal{G}_t\big],~\forall\, t\in[0,T].
\end{equation*}

Since the coefficients are supposed to be bounded and the strategies are constrained to take their values in a compact set, it is possible to solve very simply the problem by using a {\em verification} principle in terms of Lipschitz BSDEs in the spirit of that of El Karoui \emph{et al.} \cite{nekpenque97}. 

Note first that for any $\pi\in \Cc$, the process $(J_t^{\pi})$ can be easily shown to be the solution of a linear Lipschitz BSDE. More precisely, there exist $Z^{\pi}\in L^2(W)$ and $U^{\pi}\in L^2(M)$, such that $(J_t^{\pi},Z_t^{\pi},U_t^{\pi})$ is the unique solution in $\Sc^{+,\infty}\times L^2(W)\times L^2(M)$ of the linear BSDE with bounded coefficients
\begin{equation}\label{edsr jkpi}
-\,dJ^\pi_t=f^\pi(t,J_t^\pi,Z_t^\pi,U_t^\pi)dt-Z_t^\pi dW_t-U_t^\pi dM_t~;~J_T^\pi=\exp(-\gamma \xi),
\end{equation}
where $f^\pi(s,y,z,u)= \frac{\gamma^2}{2}\pi_s^2\sigma_s^2y-\gamma\pi_s(\mu_sy+\sigma_sz)-\lambda_s(1-e^{-\gamma\pi_s\beta_s})(y+u)$.

Using the fact that $J(t)=\essinf_{\pi\in\Cc_t} J_t^\pi$ for any $t\in[0,T]$, we state that $(J(t))$ corresponds to the solution of a BSDE, whose driver is the essential infimum over $\pi$ of the drivers of $(J_t^\pi)_{\pi\in\Cc}$. More precisely, 
\begin{Proposition}\label{bc}
 The following properties hold:
 \begin{itemize}
\item Let $(Y_t,Z_t,U_t)$ be the solution in $\Sc^{+,\infty}\times L^2(W)\times L^2(M)$ of the following BSDE
\begin{equation}\label{Jk edsr}
\hspace{-5mm}\left\{\begin{aligned}
 -\,dY_t=& \essinf\limits_{\pi\in \Cc}\Big\{\frac{\gamma^2}{2}\pi_t^2\sigma_t^2Y_t- \gamma\pi_t(\mu_tY_t+\sigma_tZ_t)-\lambda_t(1-e^{-\gamma\pi_t\beta_t})(Y_t+U_t)\Big\}dt  \\
&-\, Z_tdW_t-U_tdM_t , \\
Y_T=&\exp(-\gamma \xi). 
\end{aligned}\right.
\end{equation}

Then, for any $t\in[0,T]$, $J(t)=Y_t$ a.s.
\item There exists a unique optimal strategy $\hat{\pi}\in \Cc$ for $J(0)=\inf_{\pi\in\Cc}\E[\exp(-\gamma (X_T^\pi+\xi))]$ and this strategy is characterized by the fact that it attains the essential infimum in (\ref{Jk edsr}) $dt\otimes d\P-a.e.$
\end{itemize}
\label{unicite}
\end{Proposition}

\begin{proof}
Let us introduce the driver $f$ which satisfies $ds\otimes d\P-a.e.$
\begin{equation*}
f(s,y,z,u)=\essinf\limits_{\pi\in \Cc}f^\pi(s,y,z,u).
\end{equation*}
Since the driver $f$ is written as an infimum of linear drivers $f^\pi$ w.r.t $(y,z,u)$ with uniformly bounded coefficients (by assumption), $f$ is clearly Lipschitz (see Lemma \ref{fonction affine} in Appendix \ref{appendice fonction affine}). Hence, by Tang and Li's results \cite{tanli94}, BSDE (\ref{Jk edsr}) with Lipschitz driver $f$
\begin{equation*}
-\,dY_t=f(t,Y_t,Z_t,U_t)dt-Z_tdW_t-U_tdM_t~;~Y_T=\exp(-\gamma \xi)
\end{equation*}
admits a unique solution denoted by $(Y_t,Z_t,U_t)$.\\
Since, we have
\begin{equation}\label{lipsauts}
f^\pi(t,y,z,u)-f^\pi(t,y,z,u')= \lambda_t(u-u')\gamma^t,
\end{equation}
with $\gamma^t=e^{-\gamma\pi_t\beta_t}-1$ and since there exist some constants $-1<C_1\leq 0$ and $0 \leq C_2$ such that $C_1\leq \gamma^t\leq C_2$, the comparison theorem in case of jumps (see for example Theorem 2.5 in Royer \cite{roy06}) can be applied and implies that $Y_t\leq J^\pi_t$, $\forall\, t\in[0,T]$ a.s. 
As this inequality is satisfied for any $\pi\in \Cc$, it is obvious that $Y_t\leq \essinf_{\pi\in\Cc} J_t^{\pi}$ a.s. Also, by applying a measurable selection theorem, one can easily show that there exists $\hat{\pi}\in \Cc$ such that $dt  \otimes d\P$-a.s.
\begin{multline*}
 \essinf\limits_{\pi\in \Cc}\Big\{ \frac{\gamma^2}{2}\pi_t^2\sigma_t^2Y_t-\gamma\pi_t(\mu_tY_t+\sigma_tZ_t)-\lambda_t(1-e^{-\gamma\pi_t\beta_t})(Y_t+U_t)\Big\}\\
 = \frac{\gamma^2}{2}\hat{\pi}_t^2\sigma_t^2Y_t-\gamma\hat{\pi}_t(\mu_tY_t+\sigma_tZ_t)-\lambda_t(1-e^{-\gamma\hat{\pi}_t\beta_t})(Y_t+U_t).
 \end{multline*}
Thus $(Y_t,Z_t,U_t)$ is a solution of BSDE (\ref{edsr jkpi}) associated with $\hat{\pi}$. Therefore by uniqueness of the solution of BSDE (\ref{edsr jkpi}), we have $Y_t=J_t^{\hat{\pi}}$, $0\leq t \leq T$ a.s. Hence, $Y_t=\essinf_{\pi \in \mathcal{C}_t} J^{\pi}_t=J_t^{\hat{\pi}}$, $\forall \, t\in[0,T]$ a.s., and $\hat{\pi}$ is an optimal strategy. It is obvious that the optimal strategy is unique because the function $x \mapsto \exp(-\gamma x)$ is strictly convex.  
\end{proof}

\begin{Remark}{\rm The proof is short and simple thanks to the {\em verification} principle of BSDEs and optimization. Note that this {\em verification} principle is similar to the one stated in the Brownian case by El Karoui \emph{et al.} \cite{nekpenque97} but needs some particular conditions on the coefficients (see (\ref{lipsauts})) due to the presence of defaults.}
\end{Remark}
\begin{Remark}\label{Morlais}{\rm
Note that this problem has already been studied by Morlais 
\cite{mor09}. By using a verification theorem similar to that of Hu \emph{et al.} \cite{huimkmul05}, she states that the logarithm of the value function is the unique solution of a quadratic BSDE. In order to obtain this characterization, she proves the existence and the uniqueness of a solution for this quadratic BSDE with jumps by using a quite sophisticated approximation method in the spirit of Kobylanski \cite{kob00}. \\
Note that by making a change of variables, the above proposition (Proposition \ref{bc}) corresponds to Morlais's result \cite{mor09}. Indeed, put 
\begin{equation*}
\left\{\begin{aligned}
y_t=&\frac{1}{\gamma}\log(Y_t),\\
z_t=&\frac{1}{\gamma}\frac{Z_t}{Y_t}, \\
u_t=&\frac{1}{\gamma}\log\Big(1+\frac{U_t}{Y_{t^-}}\Big),
\end{aligned}\right.
\end{equation*} 
it is clear that the process $(y_t,z_t,u_t)$ is the solution of the following quadratic BSDE
\begin{equation*}
-\,dy_t=g(t,z_t,u_t)dt-z_tdW_t-u_tdM_t~;~y_T=-\xi~,
\end{equation*}
where 
\begin{equation*}
g(s,z,u)=\inf_{\pi\in\Cc}\Big(\frac{\gamma}{2}\Big|\pi_s\sigma_s-\Big(z+\frac{\mu_s+\lambda_s\beta_s}{\gamma}\Big)\Big|^2+|u-\pi_s\beta_s|_{\gamma}\Big)-(\mu_s+\lambda_s\beta_s)z-\frac{|\mu_s+\lambda_s\beta_s|^2}{2\gamma},
\end{equation*}
 which corresponds exactly to Morlais's result \cite{mor09} with $|u-\pi\beta_t|_\gamma=\lambda_t \frac{\exp(\gamma(u-\pi\beta_t))-1-\gamma(u-\pi\beta_t)}{\gamma}$.
}\end{Remark}

This characterization of the value function as the solution of a Lipschitz BSDE leads to possible numerical computations of the value function (see for example Bouchard and Elie \cite{boueli08}) and of the indifference price defined via this utility maximization problem (see Section 8).

Moreover, this property will be used to state that in the non constrained case, the value function can be approximated by a sequence of Lipschitz BSDEs (see Theorem \ref{limdeux}).

\section[The non constrained case]{The non constrained case: characterization of the value function by a BSDE}\label{section generale}
\setcounter{equation}{0} \setcounter{Assumption}{0}
\setcounter{Theorem}{0} \setcounter{Proposition}{0}
\setcounter{Corollary}{0} \setcounter{Lemma}{0}
\setcounter{Definition}{0} \setcounter{Remark}{0}

In this section, the coefficients are no longer supposed to be bounded. We now study the value function in the case where the admissible strategies are no longer required to satisfy any constraints (as in the previous section). 
Since the utility function is the exponential utility function, the set of admissible strategies is not standard in the literature. 
The next subsection studies the choice of a suitable set of admissible strategies which will allow to dynamize the problem and to characterize the associated value function (and even the dynamic value function).
\subsection{The set of admissible strategies}

Recall that in the case of the power or logarithmic utility functions defined (or restricted) on $\R^+$, the admissible strategies are the ones that make the associated wealth positive. 
Since we consider the exponential utility function $U(x)=-\exp(-\gamma x)$ which is finitely valued for all $x\in \R$, 
the wealth process is no longer required to be positive. However, it is natural to consider strategies such that the associated wealth process is uniformly bounded by below (see for example Schachermayer \cite{sch01}) or even such that any increment of the wealth is bounded by below. More precisely,
\begin{Definition}\label{portefeuille admissible}{\rm
The set of admissible trading strategies $\Ac$ consists of all $\G$-predictable processes $\pi=(\pi_t)_{0\leq t\leq T}$, which satisfy $\int_0^T|\pi_t\sigma_t|^2dt+\int_0^T\lambda_t|\pi_t\beta_t|^2dt<\infty$ a.s., and such that for any $\pi$ fixed and any $s\in[0,T]$, there exists a real constant $K_{s,\pi}$  such that $X^{\pi}_t - X^{\pi}_s \geq -K_{s,\pi}$,  $s\leq t\leq T$ a.s.
}\end{Definition}
 Recall that in their paper, Delbaen \emph{et al.} \cite{deletal02} also consider the two following sets of strategies:
 \begin{itemize}
 \item
 the set $\Theta_3$ of strategies such that the wealth process is bounded,
 \item
 the set $\Theta_2$ defined by
\begin{equation*}
\Theta_2 := \Big\{\pi\,,\, \mathbb{E}\big[ \exp\big(-\gamma (X^{\pi}_T + \xi) \big)\big] < + \infty ~ \text{and } X^\pi ~\text{is a } \Q-\text{martingale for all } \Q\in \P_f\Big\},
\end{equation*} 
where $\P_f$ is the set of absolutely continuous local martingale measures $\Q$ such that its entropy $H(\P|\Q)$ is finite.   
\end{itemize}
Note that $\Theta_3 \subset \Ac$. Of course, there is no existence result neither for the space $\Theta_3$ nor for $\Ac$ whereas there is one on the set $\Theta_2$ stated by Delbaen \emph{et al.} \cite{deletal02}. More precisely, 
by using the dual approach, under the assumption that the price process is locally bounded,
these authors show the existence of an optimal strategy on the set $\Theta_2$. \\
Also, they stress on the following important point: under the assumption that the price process is locally bounded (which is satisfied if for example $\beta$ is bounded), the value function associated with $\Theta_2$ coincides with that associated with $\Theta_3$. From this, we easily derive that these value functions also coincide with that associated with $\Ac$. More precisely,
\begin{Lemma}\label{egalite en 0} Suppose that the process $(\beta_t)$ is bounded.
The value function $V(0, \xi)$ associated with $\Ac$ defined by 
\begin{equation}\label{pb exponentielbis}
V(0, \xi)=-\inf\limits_{\pi \in \Ac} \mathbb{E}\big[\exp\big(-\gamma (X_T^{\pi}+\xi)\big)\big].
\end{equation}
is equal to the one associated with $\Theta_2$ (and also the one associated with $\Theta_3$).
\end{Lemma}

 \begin{proof}
By the result of Delbaen \emph{et al.} \cite{deletal02}, the value function associated with $\Theta_2$ coincides with that associated with $\Theta_3$ denoted by $V^3(0, \xi)$. Now, since $\Theta_3 \subset \Ac$, we have $V(0,\xi) \geq V^3(0,\xi)$. 
By a localization argument (such as in the proof of Lemma \ref{loca}), one can easily show the equality, which gives the desired result.
\end{proof}
Our aim is mainly to characterize and even to compute or approximate the value function $V(0, \xi)$.

Our approach consists in giving a dynamic extension of the optimization problem and in using stochastic calculus technics in order to characterize the dynamic value function. In the compact case (with the set $\Cc$), the dynamic extension was easy (see Section 3). At any initial time $t$, the corresponding set $\Cc_t$ of admissible strategies was simply given by the set of the restrictions to $[t,T]$ of the strategies of $\Cc$. In the case of $\Ac$ or $\Theta_3$, it is also very simple (see below for $\Ac$). However, in the case of the set $\Theta_2$, things are not so clear. Actually, this is partly linked to the fact that, contrary to the set $\Theta_2$, 
the set $\Ac$ is closed by {\em binding}. More precisely, we clearly have:

\begin{Lemma}
\label{recollement} The set $\Ac$ is {\em closed by \emph{binding}} that is:
if $\pi^1, \pi^2$ are two strategies of $\Ac$ and if $s\in [0,T]$, then the strategy $\pi^3$ defined by
\begin{equation*}
\pi^3_t=
\begin{cases}
\pi^1_t& \text{if } t\leq s,\\
\pi^2_t& \text{if } t> s,
\end{cases}
\end{equation*}
belongs to $\Ac$.
\end{Lemma}

Also, the set $\Theta_2$ is clearly {\em not closed by binding} because of the integrability condition 
$\mathbb{E}[ \exp(-\gamma (X^{\pi}_T + \xi))] < + \infty$. One could naturally think of considering the space $\Theta^{'}_2 := \{\pi\,,\, \,\,\, X^\pi ~\text{is a } \Q-\text{martingale for all } \Q\in \P_f \}$ (instead of $\Theta_2$) but this set is not really appropriate: in particular it does not allow to obtain the dynamic programming principle since the Lebesgue theorem cannot be applied (see Remark \ref{int}).

However, there are some other possible sets which are closed by binding as for example
\begin{itemize}
 \item
 the set $\Theta_3$ of strategies such that the wealth process is bounded,
 \item
the set ${\cal A}^{'}$ defined as the set of $\G$-predictable processes $\pi=(\pi_t)_{0\leq t\leq T}$ with $\int_0^T|\pi_t\sigma_t|^2dt+\int_0^T\lambda_t|\pi_t\beta_t|^2dt<\infty$ a.s., and such that for any $t\in[0,T]$ and  for any $p>1$, the following integrability condition
\begin{equation}\label{p-integrability}
\mathbb{E}\Big[\sup_{s\in[t, T]} \exp\Big(-\gamma pX^{t,\pi}_s\Big)\Big] < \infty
\end{equation}
holds. 
\end{itemize}
Note that $ \Theta_3 \subset \Ac \subset  \Ac^{'} $. 

\begin{Remark}{\rm
Note that in general, there is no existence result for the set $\Ac^{'}$.\\ 
For the proof of the closedness by binding of the set $\Ac ^{'}$ one is referred to Appendix \ref{Theta}. Note that in this proof, we see that the integrability condition 
$\mathbb{E}[ \exp(-\gamma (X^{\pi}_T + \xi))] < + \infty$ is not sufficient to derive this closedness property by binding. It is the assumption of $p$-integrability (\ref{p-integrability}) for $p>1$ (and not only the integrability)  which allows to derive the desired property. Note that this type of $p$-exponential integrability condition 
appears in some papers related to quadratic BSDEs.
}\end{Remark}


Let us now give a {\em dynamic extension} of the initial problem associated with $\Ac$ given by (\ref{pb exponentielbis}).
 For any initial time $t\in [0,T]$, we define the value function $J(t,\xi)$ by the following random variable
\begin{equation}\label{iun}
J(t,\xi)= \essinf \limits_{\pi \in \mathcal{A}_t} \mathbb{E}\big[\exp\big(-\gamma (X_T^{t,\pi}+\xi)\big)\big| \mathcal{G}_t\big],
\end{equation}
where the set $\Ac_t$ consists of all $\G$-predictable processes $\pi=(\pi_s)_{t\leq s\leq T}$, which satisfy $\int_t^T|\pi_s\sigma_s|^2ds+\int_t^T\lambda_s|\pi_s\beta_s|^2ds<\infty$ a.s., and such that for any $\pi$ fixed and any $s \in[t,T]$ there exists a constant $K_{s,\pi}$ such that 
 $X^{s,\pi}_u \geq -K_{s,\pi}$ , $s\leq u \leq T$ a.s.\\
 Note that $J(0,\xi)= -V(0, \xi)$. Also, for any $t \in [0, T]$, $J(t,\xi)$ is also equal a.s.\, to the essinf in (\ref{iun}) but taken over $\mathcal{A}$ instead of $\mathcal{A}_t$. This clearly follows from the fact that the set $\mathcal{A}_t$ is equal to the set of the restrictions to $[t,T]$ of the strategies of $\mathcal{A}$.\\
For the sake of brevity, we shall denote $J(t)$ instead of $J(t,\xi)$. Note that the random variable $J(t)$ is defined uniquely only up to $\P$-almost sure equivalent. The process $(J(t))$ will be called the {\em dynamic value function}. This process is adapted but not necessarily c\`ad-l\`ag and not even progressive.\\
Similarly, a dynamic extension of the value function associated with $\Ac'$ (or also $\Theta_3$) can be easily given. Under the assumption that the price process is locally bounded (which is satisfied if for example $\beta$ is bounded), the corresponding value functions can be easily shown to coincide a.s. More precisely, 
\begin{Lemma}\label{loca} Suppose that the coefficient $(\beta_t)$ is bounded.
The dynamic value function $(J(t))$ associated with $\Ac$ coincides a.s.\, with the one associated with $\Ac'$ (or also $\Theta_3$).
\end{Lemma}   

\begin{proof} We give here the proof for $\Ac'$ (it is the same for $\Theta_3$).
Fix $t\in[0,T]$.
Put $J^{'}(t) := \essinf_{\pi \in \mathcal{A}^{'}_t} \E[\exp(-\gamma(X^{t,\pi}_T + \xi)) | \Gc_t]$, 
where $\mathcal{A}^{'}_t$ is the set defined similarly as $\mathcal{A}^{'}$ but for initial time $t$. Note that $\mathcal{A}^{'}_t$ can be seen as the set of the restrictions to $[t,T]$ of the strategies of $\mathcal{A}^{'}$.
Since $\Ac_t \subset \mathcal{A}^{'}_t$, we get $J^{'}(t) \leq J(t)$. To prove the other inequality, we state that for any $\pi\in \mathcal{A}^{'}_t$, there exists a sequence $(\pi^n)_{n\in \N}$ of $\Ac_t$ such that $\pi^n\rightarrow \pi,~dt\otimes d\P$ a.s. Let us define $\pi^n$ by
\begin{equation*}
\pi^n_s=\pi_s\mathds{1}_{s \leq \tau_n}, ~\forall ~ s\in [t,T],
\end{equation*}
where $\tau_n$ is the stopping time defined by
$\tau_n=\inf\{s\geq t,  |X_s^{t,\pi}| \geq n\}$.\\
It is clear that for each $n\in \N$, $\pi^n\in \Ac_t$. Thus, $\exp(-\gamma X_T^{t,\pi^n})=\exp(-\gamma X_{T\wedge \tau_n}^{t,\pi})\stackrel{a.s.}{\longrightarrow} \exp(-\gamma X_T^{t,\pi})$ as $n \rightarrow +\infty$. By definition of $\mathcal{A}^{'}_t$, 
$\mathbb{E}[\sup_{s\in[t, T]} \exp(-\gamma X^{t,\pi}_s)] < \infty$. Hence, by the Lebesgue Theorem, $\E[\exp(-\gamma (X_T^{t,\pi^n}+\xi))|\Gc_t]\rightarrow \E[\exp(-\gamma (X_T^{t,\pi}+\xi))|\Gc_t]$ a.s. as $n \rightarrow +\infty$. Therefore, we have
$J(t) \leq J^{'}(t)$ a.s. which ends the proof.
\end{proof}

Hence, concerning the dynamic study of the value function, if $(\beta_t)$ is supposed to be bounded, it is equivalent to choose $\Ac$, $\Ac'$ or $\Theta_3$ as set of admissible strategies. We have chosen the set $\Ac$ because it appears as a natural set of admissible strategies from a financial point of view.

After this dynamic extension of the value function, we will use stochastic calculus technics in order to characterize the value function via a BSDE. However, it is no longer possible to use a verification theorem like the one in Section \ref{compact} because the associated BSDE is no longer Lipschitz and there is no existence result for it. One could think to use a verification theorem like that of Hu \emph{et al.} \cite{huimkmul05}. But because of the presence of jumps, it is no longer possible since again there is no existence and uniqueness results for the associated BSDE as noted by Morlais \cite{mor09}. In her paper, Morlais proves the existence of a solution of this BSDE by using an approximation method but she does not obtain uniqueness result, even in the case of bounded coefficients. Hence, this does not a priori lead to a characterization of the value function via a BSDE.

Therefore, as it seems not possible to derive a {\em sufficient condition} so that a given process corresponds to the dynamic value function, we will now directly study some properties of the dynamic value function $(J(t))$ (in other words some {\em necessary} conditions satisfied by $(J(t))$). Then, by using dynamic programming technics of stochastic control, we will derive a characterization of the value function via a BSDE. This is the object of the next section.

\subsection{Characterization of the dynamic value function as the maximal subsolution of a BSDE}
 The dynamic programming principle holds for the set ${\cal A}$:
\begin{Proposition}\label{J(t) sous martingale}
The process $(\exp(-\gamma X_t^\pi) J(t))_{0\leq t \leq T}$ is a submartingale for any $\pi \in \mathcal{A}$.
\end{Proposition}
To prove this proposition, we use the random variable $J^\pi_t$ which is defined by
\begin{equation*}
J^\pi_t=\mathbb{E}\big[\exp\big(-\gamma (X_T^{t,\pi}+\xi)\big) \big| \mathcal{G}_t\big].
\end{equation*}
As usual, in order to prove the dynamic programming principle, we first state the following lemma:

\begin{Lemma}
 The set $\{J^\pi_t,~\pi\in\mathcal{A}_t\}$ is {\em stable by pairwise minimization} for any $t\in [0,T]$. That is, for every  $\pi^1,~\pi^2  \in  \mathcal{A}_t$ there exists $\pi \in \mathcal{A}_t$ such that $ J^\pi_t=J^{\pi^1}_t \wedge J^{\pi^2}_t$.\\
Also, there exists a sequence $(\pi^n)_{n \in \mathbb{N}} \in \mathcal{A}_t$ for any $t\in [0,T]$, such that
\begin{equation*}
J(t)=\lim\limits_{n \rightarrow \infty} \downarrow J^{\pi^n}_t~a.s.
\end{equation*}
\label{stable}
\end{Lemma}

\begin{proof}
Fix $t\in [0,T]$. Let us introduce the set $E=\{J^{\pi^1}_t  \leq J^{\pi^2}_t\}$ which belongs to $\Gc_t$. Let us define $\pi$ for any $s \in [t,T]$ by 
$\pi_s=\pi_s^1\mathds{1}_{E}+\pi_s^2\mathds{1}_{E^c}.~$ 
It is obvious that $\pi\in \Ac_t$, since the sum of two random variables bounded by below is bounded by below. 
By construction of $\pi$, it is clear that $J^\pi_t=J^{\pi^1}_t \wedge J^{\pi^2}_t$.\\
The second part of lemma follows by classical results on the essential infimum (see Appendix \ref{Neveu}).
\end{proof}

Let us now give the proof of Proposition \ref{J(t) sous martingale}. 
\begin{proof}
Let us show that for $t\geq s$, 
\begin{equation*}
 \mathbb{E}\big[\exp\big(-\gamma (X_t^\pi-X_s^\pi)\big) J(t) \big|\mathcal{G}_s\big] \geq J(s)~a.s. 
 \end{equation*}
 Note that $X_t^\pi-X_s^\pi= X^{s,\pi}_t$. By Lemma \ref{stable}, there exists a sequence $(\pi_n)_{n\in\N}\in\Ac_t$ such that $J(t)=\lim\limits_{n \rightarrow \infty} \downarrow J^{\pi^n}_t$ a.s. \\
Without loss of generality, we can suppose that $\pi^0 = 0$. For each $n\in \N$, we have $J^{\pi^n}_t\leq J^{\pi^0}_t \leq 1$ a.s. Moreover, the integrability property $\E[\exp(-\gamma X^{s,\pi}_t)] < \infty$ holds because $\pi \in \Ac$. This with the Lebesgue theorem give 
\begin{equation}\label{beta}
\mathbb{E}\big[ \lim\limits_{n \rightarrow \infty}  \exp (-\gamma X^{s,\pi}_t )J^{\pi^n}_t \big|\mathcal{G}_s\big] 
=\lim\limits_{n \rightarrow \infty}  \mathbb{E}\big[  \exp (-\gamma X^{s,\pi}_t )J^{\pi^n}_t \big|\mathcal{G}_s\big].
\end{equation}
 Recall that $X^{s,\pi}_t = \int_s^t\frac{\pi_u}{S_{u^-}}dS_u$. Now, we have a.s.
\begin{equation}\label{alpha}
\exp\big(-\gamma\int_s^t\frac{\pi_u}{S_{u^-}}dS_u\big)J^{\pi^n}_t=\E\Big[\exp\Big(-\gamma (\int_s^T\frac{\tilde{\pi}^n_u}{S_{u^-}}dS_u+\xi)\Big)\Big|\Gc_t\Big],
\end{equation}
where the strategy $\tilde{\pi}^n$ is defined by
\begin{equation*}
\tilde{\pi}_u^n=
\begin{cases}
 \pi_u & \text{if }  0 \leq u \leq t,  \\
 \pi_u^n & \text{if }  t< u \leq T.
\end{cases}
\end{equation*} 
Note that by the closedness property by binding (see Lemma \ref{recollement}), $\tilde{\pi}^n\in\Ac$ for each $n\in \N$. 
\noindent By (\ref{beta}) and (\ref{alpha}), we have a.s.
\begin{eqnarray*}
\mathbb{E}\Big[ \exp\Big(-\gamma\int_s^t\frac{\pi_u}{S_{u^-}}dS_u \Big) J(t)\Big|\mathcal{G}_s\Big] 
&=& \lim\limits_{n \rightarrow \infty} \E\Big[\exp\Big(-\gamma\Big(\int_s^T\frac{\tilde{\pi}^n_u}{S_{u^-}}dS_u+\xi\Big)\Big)\Big|\Gc_s\Big]\\
&=&\lim\limits_{n \rightarrow \infty} J^{\tilde{\pi}^n}_s \,\,\,\,\,\, \geq \,\,\,\,\,\, J(s),
\end{eqnarray*}
because by definition of $J(s)$, we have $J^{\tilde{\pi}^n}_s \geq J(s)$ a.s., for each $n \in\N $. 
Hence, the process $(\exp(-\gamma X_t^\pi) J(t))$ is a submartingale for any $\pi \in \mathcal{A}$.
\end{proof}

\begin{Remark}\label{int}{\rm
Note that the integrability property $\E[\exp(-\gamma X^{s,\pi}_t)] < \infty$ is essential in the proof of this property. Indeed, if it is not satisfied, equality 
(\ref{beta}) does not hold since the Lebesgue theorem cannot be applied. One could argue that the monotone convergence theorem could be used but since the limit is decreasing, it cannot be applied without an integrability condition. 
Moreover, Fatou's lemma is not relevant since it gives an inequality but not in the suitable sense.
Actually, the importance of the integrability condition is due to the fact that we study an {\em essential infimum} of positive random variables. In the case of an {\em essential supremum} of positive random variables, the dynamic programming principle holds without any integrability condition (see for example the case of the power utility function in Lim and Quenez \cite{limque10}).
}\end{Remark}

Also, the value function can easily be characterized as follows:
\begin{Proposition} \label{le plus grand}
The process $(J(t))$ is the largest $\G$-adapted process such that $(e^{-\gamma X_t^\pi} J(t))$ is a submartingale for any admissible strategy $\pi\in \Ac$ with $J(T)=\exp(-\gamma \xi)$. More precisely, if $(\hat{J}_t)$ is a $\G$-adapted process such that $(\exp(-\gamma X_t^\pi) \hat{J}_t)$ is a submartingale for any $\pi \in \mathcal{A}$ with $\hat{J}_T=\exp(-\gamma \xi)$, then we have $J(t) \geq \hat{J}_t$ a.s., for any $t \in [0,T]$.
\end{Proposition}

\begin{proof}
Fix $t \in [0,T]$. For any $\pi\in \Ac$, 
$\mathbb{E}[\exp(-\gamma X_T^\pi) \hat{J}_T|\mathcal{G}_t] \geq \exp(-\gamma X_t^\pi) \hat{J}_t$ a.s.
This implies 
\begin{equation*}
\essinf\limits_{\pi \in \mathcal{A}_t}\mathbb{E}\big[\exp\big(-\gamma (X^{t,\pi}_T+\xi)\big)\big|\mathcal{G}_t\big] \geq \hat{J}_t~a.s.,
\end{equation*}
which gives clearly that $J(t) \geq \hat{J}_t$ a.s.
\end{proof}

With this property, it is possible to show that there exists a c\`{a}d-l\`{a}g version of the value function $(J(t))$. More precisely, we have:
\begin{Proposition}
There exists a $\G$-adapted c\`{a}d-l\`{a}g process $(J_t)$ such that for any $t \in [0,T]$,
\begin{equation*}
J_t= J(t)~a.s.
\end{equation*}
\label{cadlag}
\end{Proposition}

A direct proof is given in Appendix \ref{preuve de cadlag}.

 \begin{Remark}{\rm
Note that Proposition \ref{le plus grand} can be written under the form: $(J_t)$ is the largest $\G$-adapted c\`{a}d-l\`{a}g process such that the process $(\exp(-\gamma X_t^{\pi})J_t)$ is a submartingale for any $\pi\in \Ac$ with $J_T=\exp(-\gamma \xi)$.
\label{J sous martingale}
}\end{Remark}

We now prove that the process $(J_t)$ is bounded. More precisely, we have:
\begin{Lemma}
The process $(J_t)$ verifies 
\begin{equation*}
0 \leq  J_t \leq 1,~\forall\, t\in[0,T]~a.s.
\end{equation*}
 \label{J borne}
\end{Lemma}
\begin{proof}
Fix $t\in [0,T]$. The first inequality is easy to prove, because it is obvious that
\begin{equation*}
0\leq \mathbb{E}\big[\exp\big(-\gamma (X_T^{t,\pi}+\xi)\big)\big|\Gc_t\big]~a.s.,
\end{equation*}
for any $\pi\in \Ac_t$, which implies $0\leq J_t$.\\
The second inequality is due to the fact that the strategy defined by $\pi_s=0$ for any $s \in[t,T]$ is admissible, which implies $J_t \leq \mathbb{E}[\exp(-\gamma \xi)|\Gc_t]$ a.s. As the contingent claim $\xi$ is supposed to be non negative, we have $J_t \leq 1$ a.s.
\end{proof}

\begin{Remark}{\rm
Note that if $\xi$ is only bounded by below by a real constant $-K$, then $(J_t)$ is still upper bounded but by $\exp(\gamma K)$ instead of $1$.
}\end{Remark}

In our setting, it is not possible to use the verification theorem of Section 3 or even the verification theorem of Hu \emph{et al.} \cite{huimkmul05} in the Brownian case. Using the previous characterization of the value function (see Proposition \ref{le plus grand}), we will show directly that the value function $(J_t)$ is characterized by a BSDE. Since we work in terms of \emph{necessary conditions} satisfied by the value function, the study is more technical than in the cases where a verification theorem can be applied.\\
\indent Since $(J_t)$ is a c\`ad-l\`ag submartingale and is bounded (see Lemma \ref{J borne}), and hence of class D, it admits a unique Doob-Meyer decomposition (see Dellacherie and Meyer \cite{delmey80}, Chapter 7)
\begin{equation*}
dJ_t=dm_t+dA_t,
\end{equation*}
where $(m_t)$ is a square integrable martingale and $(A_t)$ is an increasing $\G$-predictable process with $A_0=0$. From the martingale representation theorem (see Proposition \ref{theoreme representation}), the previous Doob-Meyer decomposition can be written under the form
\begin{equation}
dJ_t=Z_tdW_t+U_tdM_t+dA_t,
\label{Doob-Meyer}
\end{equation}
with $Z\in L^2(W)$ and $U\in L^2(M)$.\\
Using the dynamic programming principle, it is possible to precise the process $(A_t)$ of (\ref{Doob-Meyer}). This allows to show that the value function $(J_t)$ is a subsolution of a BSDE. For that we define the set $\Ac^2$ of the increasing adapted c\`{a}d-l\`{a}g processes $K$ such that $K_0=0$ and $\E|K_T|^2<\infty$. More precisely,

\begin{Proposition}
There exists a process $K\in \Ac^2$ such that the process $(J_t,Z_t,U_t,K_t)\in \Sc^{+,\infty}\times L^2(W)\times L^2(M)\times \Ac^2$ is a subsolution of the following BSDE
\begin{equation}
\left\{\begin{aligned}
-\,dJ_t=&~\essinf\limits_{\pi\in \mathcal{A}}\Big\{\frac{\gamma^2}{2}\pi_t^2\sigma_t^2J_t-\gamma\pi_t(\mu_tJ_t+\sigma_tZ_t)-\lambda_t(1-e^{-\gamma\pi_t\beta_t})(J_t+U_t) \Big\}dt\\
&~-\,dK_t-Z_tdW_t-U_tdM_t,  \\
J_T=~&\exp(-\gamma \xi).
\end{aligned}\right.
\label{EDSR J}
\end{equation}

\label{A egal sup}
\end{Proposition}

\begin{proof}
The proof of this proposition is based on the dynamic programming principle: the process $(\exp(-\gamma X_t^\pi)J_t)$ is a submartingale for any $\pi\in\Ac$ (see Proposition \ref{le plus grand}). First, we write the derivative of $\exp(-\gamma X^\pi_t) J_t$ under the following form
\begin{equation*}
d(e^{-\gamma X_t^\pi}J_t)=dA_t^\pi+dm_t^\pi,
\end{equation*}
\noindent with $A_0^\pi=0$ and
\begin{equation*}
\left\{\begin{aligned}
dA_t^\pi & =e^{-\gamma X_{t}^\pi}\Big[dA_t+\Big\{\frac{\gamma^2}{2}\pi_t^2\sigma_t^2J_t-\lambda_t\big(1-e^{-\gamma \pi_t\beta_t}\big)(U_t+J_t)-\gamma\pi_t(\sigma_tZ_t+\mu_tJ_{t})\Big\}dt\Big],  \\
dm^\pi_t & =e^{-\gamma X_{t^-}^\pi}\big[(Z_t-\gamma\pi_t\sigma_tJ_{t})dW_t+(U_t+(e^{-\gamma \pi_t\beta_t}-1)(U_t+J_{t^-}))dM_t\big].
\end{aligned}\right.
\end{equation*}
Since for any $\pi \in \Ac$ the process $(\exp(-\gamma X^\pi_t) J_t)$ is a submartingale, we have
\begin{equation}
dA_t\geq \esssup\limits_{\pi\in\Ac}\Big\{\lambda_t\big(1-e^{-\gamma \pi_t\beta_t}\big)(U_t+J_t)+\gamma\pi_t(\sigma_tZ_t+\mu_tJ_{t})-\frac{\gamma^2}{2}\pi_t^2\sigma_t^2J_t\Big\}dt.
\label{A}
\end{equation}
We define the process $(K_t)$ by $K_0=0$ and 
\begin{equation*}
dK_t=dA_t -  \esssup\limits_{\pi\in\Ac}\Big\{\lambda_t\big(1-e^{-\gamma \pi_t\beta_t}\big)(U_t+J_t)+\gamma\pi_t(\sigma_tZ_t+\mu_tJ_{t})-\frac{\gamma^2}{2}\pi_t^2\sigma_t^2J_t\Big\}dt.
\end{equation*}
It is clear that the process $(K_t)$ is nondecreasing from (\ref{A}). Since the strategy defined by $\pi_t=0$ for any $t\in[0,T]$ is admissible, we have
\begin{equation*}
\esssup\limits_{\pi\in\Ac}\Big\{\lambda_t\big(1-e^{-\gamma \pi_t\beta_t}\big)(U_t+J_t)+\gamma\pi_t(\sigma_tZ_t+\mu_tJ_{t})-\frac{\gamma^2}{2}\pi_t^2\sigma_t^2J_t\Big\} \geq 0.
\end{equation*}
Hence, $0\leq K_t \leq A_t$ a.s.
As $\E|A_T|^2<\infty$, we have $K\in \Ac^2$.
\noindent Thus, the Doob-Meyer decomposition (\ref{Doob-Meyer}) of $(J_t)$ can be written as follows
\begin{align*}
dJ_t=&~\esssup\limits_{\pi\in\Ac}\Big\{\lambda_t\big(1-e^{-\gamma \pi_t\beta_t}\big)(U_t+J_t)+\gamma\pi_t(\sigma_tZ_t+\mu_tJ_{t})-\frac{\gamma^2}{2}\pi_t^2\sigma_t^2J_t\Big\}dt\\
&+\,dK_t+Z_tdW_t+U_tdM_t,
\end{align*}
with $Z\in L^2(W)$, $U\in L^2(M)$ and $K\in \Ac^2$. This ends the proof.
\end{proof}

The fact that $(J_t,Z_t,U_t,K_t)$ is a subsolution of BSDE (\ref{EDSR J}) does not allow to characterize the value function, since the subsolution of BSDE (\ref{EDSR J}) is not unique. However, we have the following characterization of the value function:

\begin{Theorem} (Characterization of the value function)\label{solution maximale}\\
$(J_t,Z_t,U_t,K_t)$ is the \emph{maximal subsolution} in $\Sc^{+,\infty} \times L^2(W)\times L^2(M) \times \Ac^2$ of BSDE (\ref{EDSR J}). That is for any subsolution $(\bar{J}_t,\bar{Z}_t,\bar{U}_t,\bar{K}_t)$ of the BSDE in $\Sc^{+,\infty} \times L^2(W) \times L^2(M) \times \Ac^2$, we have $\bar{J}_t\leq J_t$, $\forall\, t\in [0,T]$ a.s.
\end{Theorem}

\begin{Remark}
If $\xi$ and the coefficients are supposed to be bounded, we will see in Section \ref{section coefficient borne} that $(J_t,Z_t,U_t)$ is the maximal solution of BSDE (\ref{EDSR J}) that is, with $K_t=0$ for any $t\in[0,T]$ (see Theorem \ref{solution maximale bis}).
\end{Remark}

\begin{proof}
Let $(\bar{J}_t,\bar{Z}_t,\bar{U}_t,\bar{K}_t)$ be a solution of (\ref{EDSR J}) in $\Sc^{+,\infty} \times L^2(W)\times L^2(M)\times \Ac^2$. Let us prove that the process $(\exp(-\gamma X_t^\pi)\bar{J}_t)$ is a submartingale for any $\pi\in \mathcal{A}$.\\
From the product rule, we can write the derivative of this process under the form
\begin{equation*}
d\big(e^{-\gamma X_t^\pi}\bar{J}_t\big)=d\bar{M}_t^\pi+d\bar{A}_t^\pi+e^{-\gamma X_t^\pi}d\bar{K}_t,
\end{equation*}
with $\bar{A}_0^\pi=0$ and 
\begin{equation*}
\left\{\begin{aligned}
d\bar{A}_t=&-\essinf\limits_{\pi\in \mathcal{A}}\Big\{ \frac{\gamma^2}{2}\pi_t^2\sigma_t^2\bar{J}_t-\gamma\pi_t(\mu_t\bar{J}_t+\sigma_t\bar{Z}_t)-\lambda_t\big(1-e^{-\gamma\pi_t\beta_t}\big)(\bar{J}_t+\bar{U}_t)\Big\}dt, \\
d\bar{A}_t^\pi=&e^{-\gamma X_{t}^\pi}\Big\{\Big[\frac{\gamma^2}{2}\pi_t^2\sigma_t^2\bar{J}_t-\gamma\pi_t(\mu_t\bar{J}_t+\sigma_t\bar{Z}_t)-\lambda_t\big(1-e^{-\gamma\pi_t\beta_t}\big)(\bar{J}_t+\bar{U}_t)\Big]dt+d\bar{A}_t\Big\}, \\
d\bar{M}_t^\pi=&e^{-\gamma X_{t^-}^\pi}\big[(\bar{Z}_t-\gamma\pi_t\sigma_t\bar{J}_{t})dW_t+\big(\bar{U}_t+(e^{-\gamma \pi_t\beta_t}-1)\big(\bar{U}_t+\bar{J}_{t^-}\big)\big)dM_t\big].
\end{aligned}\right.
\end{equation*}

Since the strategy $\pi$ is admissible, there exists a constant $C_\pi$ such that $\exp(-\gamma X_{t}^\pi) \leq C_\pi $ for any $t \in [0,T]$. With this, one can easily derive that $\E[\sup_{t \in [0,T]} \exp(-\gamma X_t^\pi)\bar{J}_t]  < + \infty$ and that $\E[\int_0^T \exp(- \gamma X^\pi_t) d \bar{K}_t] < + \infty$. It follows that 
the local martingale $(\bar{M}^\pi_t)$ is a martingale and that 
the process $(\exp(-\gamma X^\pi_{t})\bar{J}_{t})$ is a submartingale.\\
Now recall that $(J_t)$ is the largest process such that $(\exp(-\gamma X_t^\pi) J_t)$ is a submartingale for any $\pi \in \mathcal{A}$ with $J_T=\exp(-\gamma \xi)$ (see Proposition \ref{le plus grand}). Therefore, we get
\begin{equation*}
\bar{J}_t\leq J_t , ~\forall \,  t \in [0,T]~a.s.
\end{equation*}
\end{proof} 

\begin{Remark}{\rm
Note that the integrability property $\E[\sup_{t\in [0,T]} \exp(-\gamma X^\pi_t)]$ is essential in this proof.
}\end{Remark}


\section[Approximation of the value function]{The non constrained case: approximation of the value function}\label{section approximation}
\setcounter{equation}{0} \setcounter{Assumption}{0}
\setcounter{Theorem}{0} \setcounter{Proposition}{0}
\setcounter{Corollary}{0} \setcounter{Lemma}{0}
\setcounter{Definition}{0} \setcounter{Remark}{0}

In this section, we do not make any assumptions on the coefficients of the model.\\
In the following, the value function is shown to be characterized as the {\em limit of a nonincreasing sequence} of processes $(J^k_t)_{k \in \N}$ as $k$ tends to $+ \infty$ where for each $k\in \N$, $(J^k_t)$ corresponds to the value function over the set of admissible strategies which are bounded by $k$. 

Note that in the classical case of {\em bounded coefficients}, we will see in the next section that for each $k\in \N$, $(J^k_t)$ can be characterized as the solution of a {\em Lipschitz BSDE}. 
 
\vspace{2mm}
For each $k \in \N$, we denote by $\Ac^k_t$ the set of strategies of $\Ac_t$ uniformly bounded by $k$, and we consider the associated value function $J^k(t)$ defined by 
 \begin{equation}\label{defk}
J^k(t)=\essinf_{\pi\in\Ac^k_t}\E\big[\exp\big(-\gamma(X^{t,\pi}_T+\xi)\big)\big|\Gc_t\big].
\end{equation}
By similar argument as for $(J_t)$, there exists a càd-làg version of $(J^k(t))$ denoted by $(J^k_t)$.
 As previously, the dynamic programming principle holds:
\begin{Proposition}
The process $(\exp(-\gamma X_t^{\pi})J^k_t)$ is a submartingale for any $\pi\in\Ac^k$.
\label{Jk sous martingale}
\end{Proposition}

We now show that the value functions $(J^k_t)_{k\in\N}$ converge to the value function $J_t$. More precisely, we have:
\begin{Theorem}\label{J limite}  (Approximation of the value function)\\
 For any $t\in[0,T]$, we have
\begin{equation*}
J_t=\lim\limits_{k\rightarrow \infty} \downarrow J^k_t~a.s.
\end{equation*}
\end{Theorem}

\begin{proof}  Fix $t\in[0,T]$.
It is obvious with the definitions of sets $\mathcal{A}_t$ and $\mathcal{A}^k_t$ that $\mathcal{A}^k _t\subset \mathcal{A}_t$ for each $k\in \N$ and hence
\begin{equation*}
J_t\leq J^k_t~a.s.
\end{equation*}
Moreover, since $\Ac_t^k \subset \Ac_t^{k+1}$ for each $k\in \N$, it follows that the sequence of positive random variables $(J^k_t)_{k\in \mathbb{N}}$ is nonincreasing. Let us define the random variable
\begin{equation*}
\bar{J}(t)=\lim\limits_{k\rightarrow \infty} \downarrow J^k_t~a.s.
\end{equation*}
It is obvious from the previous inequality that $J_t\leq \bar{J}(t)$ a.s., and this holds for any $t\in[0,T]$.\\
 It remains to prove that $J_t\geq \bar{J}(t)$ a.s. for any $t\in[0,T]$. This will be done by the following steps.\\
\textbf{Step 1:}
Let us now prove that the process $(\bar{J}(t))$ is a submartingale. Fix $0\leq s< t\leq T$. From Proposition \ref{Jk sous martingale}, $(J^k_t)$ is a submartingale, which gives for each $k\in \N$
\begin{equation*}
 \mathbb{E}\big[J^k_t \big| \mathcal{G}_s\big] \geq J^k_s \geq \bar{J}(s)~a.s.
 \end{equation*}
The dominated convergence theorem (which can be applied since $0\leq J^k_t\leq 1$ for each $k\in \N$) gives
\begin{equation*}
\E\big[\bar{J}(t) \big| \mathcal{G}_s\big]=\lim_{k\rightarrow \infty} \E\big[J^k_t \big| \mathcal{G}_s\big] \geq \bar{J}(s)~a.s.\,,
\end{equation*}
which gives step 1.\\

\noindent \textbf{Step 2:} Let us show that the process $(\exp(-\gamma X_t^\pi)\bar{J}(t))$ is a submartingale for any bounded strategy $\pi\in \Ac$. \\
Let $\pi$ be a bounded admissible strategy. Then, there exists $n\in \N$ such that $\pi$ is uniformly bounded by $n$. For each $k\geq n$, since $\pi \in \Ac^k$, $(\exp(-\gamma X_t^\pi)J^k_t)$ is a submartingale from Proposition \ref{Jk sous martingale}. Then, by the dominated convergence theorem, the process $(\exp(-\gamma X_t^\pi)\bar{J}(t))$ can be easily proven to be a submartingale.\\

\noindent\textbf{Step 3:} Note now that the process $(\bar{J}(t))$ is a submartingale not necessarily c\`{a}d-l\`{a}g. 
However, by a theorem of Dellacherie-Meyer \cite{delmey80} (see VI.18), we know that the nonincreasing limit of a sequence of càd-làg submartingales is indistinguishable from a càd-làg adapted process. 
Hence, there exists a càd-làg version of $(\bar{J}(t))$ which will be denoted by $(\bar{J}_t)$. Note that 
$(\bar{J}_t)$ is still a submartingale. \\

\noindent \textbf{Step 4:} Let us show that $\bar{J}_t\leq J_t,~\forall\,t\in [0,T]$ a.s. Since by steps 1 and 3, $(\bar{J}_t)$ is a c\`{a}d-l\`{a}g submartingale of class D, it admits the following Doob-Meyer decomposition
\begin{equation*}
d\bar{J}_t=\bar{Z}_tdW_t+\bar{U}_tdM_t+d\bar{A}_t,
\end{equation*}
where $\bar{Z}\in L^2(W)$, $\bar{U}\in L^2(M)$ and $(\bar{A}_t)$ is a nondecreasing $\G$-predictable process with $\bar{A}_0=0$.\\
As before, we use the fact that the process $(\exp(-\gamma X_t^\pi)\bar{J}_t)$ is a submartingale for any bounded strategy $\pi\in \Ac$ to give some necessary conditions satisfied by the process $(\bar{A}_t)$.\\
Let $\pi\in \Ac$ be a uniformly bounded strategy. The product rule gives
\begin{equation*}
d(e^{-\gamma X_t^\pi}\bar{J}_t)=d\bar{M}^\pi_t+d\bar{A}_t^\pi,
\end{equation*}
with $\bar{A}_0^\pi=0$ and 
\begin{equation*}
\left\{\begin{aligned}
d\bar{A}_t^\pi & =e^{-\gamma X_{t}^\pi}\Big\{d\bar{A}_t+\Big[\frac{\gamma^2}{2}\pi_t^2\sigma_t^2\bar{J}_{t}+\lambda_t(e^{-\gamma \pi_t\beta_t}-1)(\bar{U}_t+\bar{J}_t)-\gamma\pi_t(\mu_t\bar{J}_t+\sigma_t\bar{Z}_t)\Big]dt\Big\}, \\
d\bar{M}_t^\pi & =e^{-\gamma X_{t^-}^\pi}\big[(\bar{Z}_t-\gamma\pi_t\sigma_t\bar{J}_{t})dW_t+(\bar{U}_t+(e^{-\gamma \pi_t\beta_t}-1)(\bar{U}_t+\bar{J}_{t^-}))dM_t\big].
\end{aligned}\right.
\end{equation*}
Let $\bar{\mathcal{A}}$ be the set of uniformly bounded admissible strategies. Since the process $(e^{-\gamma X_t^\pi}\bar{J}_t)$ is a submartingale for any $\pi \in \bar{\Ac}$, we have $d\bar{A}_t^\pi \geq 0$ a.s.\, for any $\pi \in \bar{\Ac}$. Hence, there exists a process $\bar{K}\in \Ac^2$ such that
 \begin{equation*}
 d\bar{A}_t = -\essinf\limits_{\pi\in \bar{\mathcal{A}}}\Big\{\frac{\gamma^2}{2}\pi_t^2\sigma_t^2\bar{J}_t-\gamma\pi_t(\mu_t\bar{J}_t+\sigma_t\bar{Z}_t)-\lambda_t(1-e^{-\gamma\pi_t\beta_t})(\bar{J}_t+\bar{U}_t)\Big\}dt+d\bar{K}_t.
 \end{equation*}
Now, the following equality holds $dt\otimes d\P-a.e.$ (see Appendix \ref{esssup A Abar} for details)
\begin{multline}\label{esssupA=esssupAbar}
\essinf\limits_{\pi\in \bar{\mathcal{A}}}\Big\{\frac{\gamma^2}{2}\pi_t^2\sigma_t^2\bar{J}_t-\gamma\pi_t(\mu_t\bar{J}_t+\sigma_t\bar{Z}_t)-\lambda_t(1-e^{-\gamma\pi_t\beta_t})(\bar{J}_t+\bar{U}_t)\Big\}=\\
\essinf\limits_{\pi\in \mathcal{A}}\Big\{\frac{\gamma^2}{2}\pi_t^2\sigma_t^2\bar{J}_t-\gamma\pi_t(\mu_t\bar{J}_t+\sigma_t\bar{Z}_t)-\lambda_t(1-e^{-\gamma\pi_t\beta_t})(\bar{J}_t+\bar{U}_t) \Big\}.
\end{multline}
Hence, $(\bar{J}_t,\bar{Z}_t,\bar{U}_t,\bar{K}_t)$ is a subsolution of BSDE (\ref{EDSR J}) and Theorem \ref{solution maximale} implies that
\begin{equation*}
\bar{J}_t\leq J_t,~\forall\,t\in[0,T]~a.s.,
\end{equation*}
which ends the proof.
\end{proof}

In the next section, we will see that in the classical case of {\em bounded coefficients}, for each $k\in \N$, $(J^k_t)$ can be characterized as the solution of a {\em Lipschitz BSDE}.

\section{Case of bounded coefficients}\label{section coefficient borne}
\setcounter{equation}{0} \setcounter{Assumption}{0}
\setcounter{Theorem}{0} \setcounter{Proposition}{0}
\setcounter{Corollary}{0} \setcounter{Lemma}{0}
\setcounter{Definition}{0} \setcounter{Remark}{0}

In this section, the coefficients of the model $(\mu_t)$, $(\sigma_t)$, $(\beta_t)$ and $(\lambda_t)$ are supposed to be bounded. We will see that in this case, the two previous theorems (Theorem \ref{solution maximale} and Theorem \ref{J limite}) will lead to more precise characterizations of the dynamic value function.

For each $k \in \N$, we define the set $\Bc^k$ as the set of all strategies (not necessarily in $\Ac$) such that they take their values in $[-k,k]$. Also, we denote
by $\Bc^k_t$ the set of all strategies beginning at $t$ and such that they take their values in $[-k,k]$.\\

Note that for each $k\in \N$, $\forall\, p>1$ and $\forall\,t \in [0,T]$ the following integrability property
\begin{equation}\label{lemme integrabilite 2}
\sup_{\pi \in \Bc^k} \E\big[\exp(-\gamma p X^\pi_t)\big] < \infty
\end{equation}
clearly holds.

We state the following lemma:
\begin{Lemma}\label{cinq}
The following equality holds for any $k\in \N$ and for any $t\in [0,T]$
\begin{equation*}
{J}^k_t = \essinf_{\pi \in \Bc^k_t} \E\big[\exp(-\gamma(X^{t,\pi}_T + \xi)) \big| \Gc_t\big]~a.s.,
\end{equation*}
with $({J}^k_t)$ defined in the previous section by (\ref{defk}).
\end{Lemma}

\begin{proof} Fix $k\in \N$ and $t\in[0,T]$.
Put $\bar{J}^k_t := \essinf_{\pi \in \Bc^k_t} \E[\exp(-\gamma(X^{t,\pi}_T + \xi)) | \Gc_t]$.
Since $\Ac^k_t \subset \Bc^k_t$, we get $\bar{J}^k_t \leq J^k_t$. To prove the other inequality, we state that there exists a sequence $(\pi^n)_{n\in \N}$ of $\Ac^k_t$ such that $\pi^n\rightarrow \pi,~dt\otimes d\P$ a.s., for any $\pi\in\Bc^k_t$. Let us define $\pi^n$ by
\begin{equation*}
\pi^n_s=\pi_s\mathds{1}_{s \leq \tau_n}, ~\forall ~ s\in [t,T],
\end{equation*}
where $\tau_n$ is the stopping time defined by
$\tau_n=\inf\{s\geq t,  |X_s^{t,\pi}| \geq n\}$.\\
It is clear that for each $n\in \N$, $\pi^n\in \Ac^k_t$. Thus, $\exp(-\gamma X_T^{t,\pi^n})=\exp(-\gamma X_{T\wedge \tau_n}^{t,\pi})\stackrel{a.s.}{\longrightarrow} \exp(-\gamma X_T^{t,\pi})$ as $n \rightarrow +\infty$. By (\ref{lemme integrabilite 2}), the set of random variables $\{\exp(-\gamma X_T^{t,\pi}),~\pi \in \Bc^k_t\}$ is uniformly integrable. Hence, $\E[\exp(-\gamma (X_T^{t,\pi^n}+\xi))|\Gc_t]\rightarrow \E[\exp(-\gamma (X_T^{t,\pi}+\xi))|\Gc_t]$ a.s. as $n \rightarrow +\infty$. Therefore, we have
$J^k_t \leq \bar{J}^k_t$ a.s. which ends the proof.
\end{proof}

Now by Proposition \ref{bc}, we have that for each $k \in \N$, the process $(J^k_t)$ is characterized as the solution of a 
{\em Lipschitz} BSDE given by (\ref{Jk edsr}) with $\Cc$ replaced by $\Bc^k$. Hence, we have that:
\begin{Theorem} \label{lim} (Approximation of the value function)\\
The value function is characterized as the {\em nonincreasing limit} of the sequence $(J^k_t)_{k \in \N}$ as $k$ tends to $+ \infty$, where for each $k$, $(J^k_t)$ is the solution of {\em Lipschitz BSDE} (\ref{Jk edsr}) with $\Cc=
\Bc^k$.
\end{Theorem}
\begin{Remark}{\rm
Note that this allows to {\em approximate} the value function by {\em numerical computations} (by applying for example Bouchard and Elie's results \cite{boueli08}). 
}\end{Remark}

We now recall a result of convergence stated by Morlais \cite{mor09}. 
For each $k \in \N$, let us denote by $(Z^k_t, U^k_t)$ the pair of square integrable processes such that 
$(J^k_t,Z^k_t, U^k_t)$ is solution of the associated {\em Lipschitz} BSDE (\ref{Jk edsr}) with $\Cc$ replaced by $\Bc^k$. We make the following change of variables
\begin{equation*}
\left\{\begin{aligned}
y^k_t=&\frac{1}{\gamma}\log(J^k_t),\\
z^k_t=&\frac{1}{\gamma}\frac{Z^k_t}{J^k_t}, \\
u^k_t=&\frac{1}{\gamma}\log\Big(1+\frac{U^k_t}{J^k_{t^-}}\Big).
\end{aligned}\right.
\end{equation*} 
It is clear that the process $(y^k_t,z^k_t,u^k_t)$ is a solution of the following quadratic BSDE
\begin{equation*}
-\,dy^k_t=g^k(t,z^k_t,u^k_t)dt-z^k_tdW_t-u^k_tdM_t~;~y^k_T=-\xi~,
\end{equation*}
where 
\[g^k(s,z,u)=\inf_{\pi\in\Bc^k}\Big(\frac{\gamma}{2}\Big|\pi_s\sigma_s-\Big(z+\frac{\mu_s+\lambda_s\beta_s}{\gamma}\Big)\Big|^2+|u-\pi_s\beta_s|_{\gamma}\Big)-(\mu_s+\lambda_s\beta_s)z-\frac{|\mu_s+\lambda_s\beta_s|^2}{2\gamma}\]
 with $|u-\pi\beta_t|_\gamma=\lambda_t\frac{\exp(\gamma(u-\pi\beta_t))-1-\gamma(u-\pi\beta_t)}{\gamma}$.\\
Recall that by using Kobylanski's technics \cite{kob00} on monotone stability convergence theorem, Morlais \cite{mor09} shows the following nice result:

\begin{Proposition} (Morlais's result) \label{morlaisbis}
Suppose that the coefficients are bounded and that $\xi$ is bounded. Then, $(y^k_t,z^k_t,u^k_t)$ converges to $(y_t,z_t,u_t)$ in the following sense
\begin{equation*}
\E(\sup_{t\in[0,T]}|y^k_t-y_t|)+|z^k-z|_{L^2(W)}+|u^k-u|_{L^2(M)}\rightarrow 0~,
\end{equation*}
where $(y_t,z_t,u_t)$ is solution of
\begin{equation*}
-\,dy_t=g(t,y_t,z_t,u_t)dt-z_tdW_t-u_tdM_t~;~y_T=-\xi~,
\end{equation*}
with 
\[g(s,z,u)=\inf_{\pi\in \bar \Bc}\Big(\frac{\gamma}{2}\Big|\pi_s\sigma_s-\Big(z+\frac{\mu_s+\lambda_s\beta_s}{\gamma}\Big)\Big|^2+|u-\pi_s\beta_s|_{\gamma}\Big)-(\mu_s+\lambda_s\beta_s)z-\frac{|\mu_s+\lambda_s\beta_s|^2}{2\gamma},\]
where $\bar \Bc= \cup_k \Bc^k$.
\end{Proposition}

By similar arguments as in the proof of the above lemma (Lemma \ref{cinq}) or as in Appendix \ref{esssup A Abar}, the set $\bar \Bc$ can be replaced by $\bar \Ac$ or even by $\Ac$.

Using this proposition and our characterization of $(J_t)$ as the {\em nonincreasing limit} of $(J^k_t)_{k \in \N}$, we can identify the limit 
$(y_t)$. More precisely, let us define the following processes
\begin{equation*}
\left\{\begin{aligned}
J^*_t&=e^{\gamma y_t}, \\
Z^*_t&=\gamma J^*_tz_t,  \\
U^*_t&=(e^{\gamma u_t}-1)J^*_{t^-}.
\end{aligned}\right.
\end{equation*}
Since $J_t=\lim_{k\rightarrow \infty} J^k_t$ by Theorem \ref{lim} (or \ref{J limite}), $J^*_t=J_t, ~\forall\, t\in[0,T]$ a.s., and the uniqueness of the Doob-Meyer decomposition (\ref{Doob-Meyer}) of $J_t$ implies that $Z^*_t=Z_t$ and $U^*_t=U_t$ $dt\otimes d\P-a.e.$ Also, by using Morlais's result (Proposition \ref{morlaisbis}), we derive that $(J_t,Z_t,U_t)$ is a solution of BSDE (\ref{EDSR J}), and not only a subsolution.
This, with the characterization of $(J_t)$ of Theorem \ref{solution maximale}, give:

\begin{Theorem}\label{solution maximale bis}
(Characterization of the value function)\\
Suppose that $\xi$ and the coefficients are bounded. Then, the value function $(J_t,Z_t,U_t)$ is the maximal solution of BSDE (\ref{EDSR J}) (that is with $K_t=0$ for any $t\in[0,T]$).
\end{Theorem}

\begin{Remark}\label{egale A'} {\rm
Moreover, if there is no default, our result corresponds to that of Hu \emph{et al.} \cite{huimkmul05} in the complete case (by making the simple exponential change of variable $y_t=\frac{1}{\gamma}\log(J_t)$).
Also, in this case, the optimal strategy belongs to the set $\Ac'$. Indeed, the optimal terminal wealth is given by $\hat X_T = I(\lambda Z_0 (T))$, where $I$ is the inverse of $U^{'}$, $\lambda$ is a fixed parameter, $Z_0(T) := \exp\{-\int_0^T \alpha_t dW_t - \frac{1}{2} \int_0^T \alpha_t ^2 dt\}$ and $\alpha_t := \frac{\mu_t+ \lambda_t \beta_t}{\sigma_t}$ (supposed to be bounded). 
}\end{Remark}

\section[Coefficients satisfying some integrability conditions]{Case of coefficients which satisfy some exponential integrability conditions}\label{section coefficient non borne}
\setcounter{equation}{0} \setcounter{Assumption}{0}
\setcounter{Theorem}{0} \setcounter{Proposition}{0}
\setcounter{Corollary}{0} \setcounter{Lemma}{0}
\setcounter{Definition}{0} \setcounter{Remark}{0}

In this section, we will study the case of coefficients not necessarily bounded but satisfying some integrability conditions.
We will first study the particular case of strategies valued in a convex-compact set. Then, we will see that the approximation result of the value function in the non constrained case (Theorem \ref{J limite}) can be precised.

\subsection{Case of strategies valued in a convex-compact set}
We suppose that the set of admissible strategies is given by $\Cc$ (see Section \ref{compact}) where $C$ is a convex-compact (not only compact) set. Here, it simply corresponds to a closed interval of $\R$ because we are in the one dimensional case. However, the following results clearly still hold in the multidimensional case (see Section 9).
Let $(J(t))$ be the associated dynamic value function to $\Cc_t$ defined as in Section 3 (see (\ref{Jkt})). Using some classical results of convex analysis (see for example Ekeland and Temam \cite{eketem76}), we easily state the following existence property:
\begin{Proposition}
\label{0-optimal}
There exists an optimal strategy $\hat{\pi}\in\Cc$ for the optimization problem (\ref{pb exponentiel}), that is
\begin{equation*}
J(0)= \inf \limits_{\pi \in \Cc} \mathbb{E}\big[\exp\big(-\gamma (X_T^{\pi}+\xi)\big)\big]=\mathbb{E}\big[\exp\big(-\gamma (X_T^{\hat{\pi}}+\xi)\big)\big].
\end{equation*}
\end{Proposition}

\begin{proof}
Note that $\Cc$ is strongly closed and convex in $L^2([0,T]\times \Omega)$. Hence, $\Cc$ is closed for the weak topology. Moreover, since $\Cc$ is bounded, $\Cc$ is compact for the weak topology.\\
We define the function $\phi(\pi)=\E[\exp(-\gamma (X_T^\pi+\xi))]$ on $L^2([0,T]\times \Omega)$. This function is clearly convex and continuous for the strong topology in $L^2([0,T]\times \Omega)$. By classical results of convex analysis, it is s.c.i for the weak topology. Now, there exists a sequence $(\pi^n)_{n\in \N}$ of $\Cc$ such that $\phi(\pi^n)\rightarrow \min_{\pi\in\Cc}\phi(\pi)$ as $n\rightarrow \infty$.\\
Since $\Cc$ is weakly compact, there exists an extracted sequence still denoted by $(\pi^{n})$ which converges for the weak topology to $\hat{\pi}$ for some $\hat{\pi}\in \Cc$. Now, since $\phi$ is s.c.i for the weak topology, it implies that
\begin{equation*}
\phi(\hat{\pi})\leq \lim\inf \phi(\pi^n)=\min_{\pi\in\Cc}\phi(\pi).
\end{equation*}
Therefore, $\phi(\hat{\pi})=\inf_{\pi\in\Cc}\phi(\pi)$ and the proof is ended.
 \end{proof}
 
 We now want to characterize the value function $J(t)$ by a BSDE. For that we cannot apply the same technics as in the case of bounded coefficients. Indeed, since the coefficients are not necessarily bounded, the drivers of the associated BSDEs are no longer Lipschitz. Hence, the existence and uniqueness results and also the comparison theorem do not a priori hold. 
Therefore, as in Section \ref{section generale}, we will use dynamic programming technics of stochastic control but also the existence of an optimal strategy.

First, one can show easily that the set $\{J^\pi_t,~\pi \in \Cc_t\}$ is {\em stable by pairwise minimization}.\\
In order to have the dynamic programming principle, we now suppose that the coefficients satisfy the following integrability condition: 
\begin{Assumption} \label{condition integrabilite}
$(\beta_t)$ is uniformly bounded and 
\begin{equation*}
\E\Big[\exp\Big(a \int_0^T|\mu_t|dt\Big)\Big]+\E\Big[\exp\Big(b \int_0^T|\sigma_t|^2dt\Big)\Big]<\infty,
\end{equation*}
where $a= 2 \gamma ||\Cc||_\infty$ and $b= 8\gamma^2||\Cc||_\infty^2$.
\end{Assumption}

By classical computations, one can easily derive that for any $t \in [0,T]$ and any $\pi \in \Cc_t$, the following inequality holds
\begin{equation}\label{lemme integrabilite}
\E\big[ \sup_{s\in [t,T]} \exp\big(-\gamma X^{t,\pi}_s\big) \big] < \infty.
\end{equation}

Using this integrability property and similar arguments as in the proof of Proposition \ref{J(t) sous martingale}, the process $(J(t))$ can be shown to satisfy the {\em dynamic programming principle} over $\Cc$ that is: $(J(t))$ is the largest $\G$-adapted process such that $(\exp(-\gamma X_t^\pi) J(t))$ is a submartingale for any $\pi\in \Cc$ with $J(T)=\exp(-\gamma \xi)$. 

Also, the following classical optimality criterion holds:
\begin{Proposition}
Let $\hat{\pi}\in \Cc $. The two following assertions are equivalent:
\begin{enumerate}[(i)]
\item $\hat{\pi}\in \Cc $ is optimal that is $J (0)=\E[\exp(-\gamma(X_T^{\hat{\pi}}+\xi))]$ 
\item The process $(\exp(-\gamma X_t^{\hat{\pi}})J (t))$ is a martingale.
\end{enumerate}
\label{J martingale}
\end{Proposition}
The proof is given in Appendix \ref{Proof of Proposition J martingale}.

 \begin{Corollary}
 There exists a c\`{a}d-l\`{a}g version of $(J (t))$ which will be denoted by $(J _t)$.
\end{Corollary}

\begin{proof}
The proof is simple here because we have an existence result. More precisely, from Proposition \ref{0-optimal}, there exists $\hat{\pi}\in\Cc $ which is optimal for $J _0$. Hence, by the optimality criterium (Proposition \ref{J martingale}), we have $J (t)= \exp(-\gamma X_t^{\hat{\pi}})\E[\exp(-\gamma(X_T^{\hat{\pi}}+\xi))|\Gc_t]$ for any $t\in[0,T]$ (in other words, $\hat{\pi}$ is also optimal for $J (t)$). By classical results on the conditional expectation, there exists a c\`{a}d-l\`{a}g version denoted by $(J _t)$.
\end{proof}

Note that the process $(J _t)$ verifies
$0\leq J _t \leq 1,~\forall\,t\in[0,T]$ a.s. Using the dynamic programming principle and the existence of an optimal strategy, we state the following property:
\begin{Proposition}\label{proposition Jk}
There exist $Z \in L^2(W)$ and $U \in L^2(M)$ such that $(J _t,Z _t,U _t)$ is the maximal solution in $\Sc^{+,\infty} \times  L^2(W) \times  L^2(M)$ of BSDE (\ref{Jk edsr}).
\end{Proposition}

The proof is given in Appendix \ref{appendice edsr}.
\begin{Remark}{\rm
It can be noted that the optimal strategy $\hat{\pi}\in \Cc $ for $J _0$ is characterized by the fact that $\hat{\pi}_t$ attains the essential infimum in (\ref{Jk edsr}), $dt\otimes d\P-a.e.$
}\end{Remark}

With Assumption \ref{condition integrabilite} it is possible to prove the unicity of the solution to BSDE (\ref{Jk edsr}).
\begin{Theorem}  \label{proposition unique solution}(Characterization of the value function)\\
The value function $(J_t,Z _t,U _t)$ is characterized as the \emph{unique solution} in $\Sc^{+,\infty} \times L^2(W)\times L^2(M)$ of BSDE (\ref{Jk edsr}). 
\end{Theorem}

\begin{proof}
Let $(\bar{J}_t,\bar{Z}_t,\bar{U}_t)$ be a solution of BSDE (\ref{Jk edsr}). Using a measurable selection theorem, we know that there exists at least a strategy $ \bar{\pi}\in\Cc $ such that $dt\otimes d\P-a.e.$
\begin{multline*}
\essinf_{\pi\in\Cc }\Big\{\frac{\gamma^2}{2}\pi_t^2\sigma_t^2\bar{J}_t-\gamma\pi_t(\mu_t\bar{J}_t+\sigma_t\bar{Z}_t)-\lambda_t(1-e^{-\gamma\pi_t\beta_t})(\bar{J}_t+\bar{U}_t)\Big\}\\
=\frac{\gamma^2}{2} \bar{\pi}^2\sigma_t^2\bar{J}_t-\gamma \bar{\pi}_t(\mu_t\bar{J}_t+\sigma_t\bar{Z}_t)-\lambda_t(1-e^{-\gamma \bar{\pi}_t\beta_t})(\bar{J}_t+\bar{U}_t).
\end{multline*}
Thus (\ref{Jk edsr}) can be written under the form 
\begin{equation*}
d\bar{J}_t=\Big\{\gamma \bar{\pi}_t(\mu_t\bar{J}_t+\sigma_t\bar{Z}_t)+\lambda_t(1-e^{-\gamma \bar{\pi}_t\beta_t})(\bar{J}_t+\bar{U}_t)-\frac{\gamma^2}{2} \bar{\pi}^2\sigma_t^2\bar{J}_t\Big\}dt+\bar{Z}_tdW_t+\bar{U}_tdM_t.
\end{equation*}
Let us introduce by $B_t=\exp(-\gamma X_t^{ \bar{\pi}})$. It\^{o}'s formula and rule product give
\begin{equation*}
d(B_t\bar{J}_t)=\big(B_t\bar{Z}_t-\gamma \sigma_t \bar{\pi}_tB_t\bar{J}_t\big)dW_t+\big[(e^{-\gamma \beta_t \bar{\pi}_t}-1)B_{t^-}\bar{J}_t+e^{-\gamma \beta_t \bar{\pi}_t}B_{t^-}\bar{U}_t\big]dM_t.
\end{equation*} 
By Assumption \ref{condition integrabilite} and since $(\bar{J}_t)$ is bounded, one can derive that the local martingale $(B_t\bar{J}_t)$ satisfies $\E[\sup_{0\leq t \leq T} |B_t\bar{J}_t|] < \infty$. Hence, $(B_t\bar{J}_t)$ is a martingale. Thus, 
\begin{equation*}
\bar{J}_t=\E\Big[\frac{B_T}{B_t}e^{-\gamma\xi}\Big|\Gc_t\Big]=\E\big[\exp(-\gamma(X_T^{t, \bar{\pi}}+\xi))\big|\Gc_t\big].
\end{equation*}
Thus,
\begin{equation*}
\bar{J}_t\geq \essinf_{\pi\in \Cc }\E\big[\exp(-\gamma(X_T^{t,\pi}+\xi))\big|\Gc_t\big]=J_t .
\end{equation*}
Now, by the previous Proposition \ref{proposition Jk}, $(J_t )$ is the maximal solution of BSDE (\ref{Jk edsr}). This gives that for any $t\in [0,T]$, $J_t \leq \bar{J}_t,~\P-a.s.$ Hence, $J_t = \bar{J}_t,~\forall \, t\in [0,T],~ \P-a.s.$, and $\bar \pi$ is optimal and the proof is ended.

\end{proof}

\subsection{The non constrained case}
In this section, the set of admissible strategies is given by ${\cal A}$.
Under some exponential integrability conditions on the coefficients, we can also precise the characterization of the value function 
$(J_t)$ as the limit of $(J^k_t)_{k\in \N}$ as $k$ tends to $+ \infty$.

\begin{Assumption}\label{hypo integrabilite}{\rm
$(\beta_t)$ is uniformly bounded, $\E[\int_0^T\lambda_tdt]<\infty$ and for any $p>0$ we have
\begin{equation*}
\E\Big[\exp\Big(p\int_0^T|\mu_t|dt\Big)\Big]+\E\Big[\exp\Big(p\int_0^T|\sigma_t|^2dt\Big)\Big]<\infty.
\end{equation*}
}\end{Assumption}
Again, for each $k \in \N$, we consider the set $\Bc^k_t$ of strategies beginning at $t$ and valued in $[-k,k]$. Since Assumption \ref{hypo integrabilite} is satisfied, the integrability condition (\ref{lemme integrabilite 2}) holds and hence, for each $k\in \N$,
\begin{equation*}
{J}^k_t = \essinf_{\pi \in \Bc^k_t} \E\big[\exp\big(-\gamma(X^{t,\pi}_T + \xi)\big) \big| \Gc_t\big]~a.s.
\end{equation*}
In this case, for each $k\in \N$, the process $(J^k_t)$ is characterized as the unique solution of BSDE (\ref{Jk edsr}) with $\Cc=\Bc^k$. Therefore, we have:

\begin{Theorem} \label{limdeux} (Characterization of the value function)\\
The value function is characterized as the {\em nonincreasing limit} of the sequence $(J^k_t)_{k\in \N}$ as $k$ tends to $+ \infty$, which are {\em the unique solutions} of BSDEs (\ref{Jk edsr}) with $\Cc= \Bc^k$ for each $k \in \N$.
\end{Theorem}

\section[Indifference pricing]{Indifference pricing via the maximization of exponential utility}
\label{indifference pricing}
\setcounter{equation}{0} \setcounter{Assumption}{0}
\setcounter{Theorem}{0} \setcounter{Proposition}{0}
\setcounter{Corollary}{0} \setcounter{Lemma}{0}
\setcounter{Definition}{0} \setcounter{Remark}{0}

We first present a general framework of the Hodges and Neuberger \cite{hodneu89} approach with some strictly increasing, strictly concave and continuously differentiable mapping $U$, defined on $\R$. We solve explicitly the problem in the case of exponential utility. \\
\indent The Hodges approach to pricing of unhedgeable claims is a utility-based approach and can be summarized as follows the issue at hand is to assess the value of some (defaultable) claim $\xi$ as seen from the perspective of an investor who optimizes his behavior relative to some utility function, say $U$. The investor has two choices 
\begin{itemize}
\item he only invests in the risk-free asset and in the risky asset, in this case the associated optimization problem is
\begin{equation*}
V(x,0)=\sup\limits_{\pi}\E\big[U(X_T^{x,\pi})\big],
\end{equation*}
\item he also invests in the contingent claim, whose price is $p$ at $0$, in this case the associated optimization problem is
\begin{equation*}
V(x-p,\xi)=\sup\limits_{\pi}\E\big[U(X_T^{x-p,\pi}+\xi)\big].
\end{equation*}
\end{itemize}  
\begin{Definition}
{\rm For a given initial endowment $x$, the Hodges buying price of a defaultable claim $\xi$ is the price $p$ such that the investor's value functions are indifferent between holding and not holding the contingent claim $\xi$, i.e.
\begin{equation*}
V(x,0)=V(x-p,\xi).
\end{equation*}
}\end{Definition}

\begin{Remark}{\rm
We can define the Hodges selling price $p_*$ of $\xi$ by considering $-p$, where $p$ is the buying price of $-\xi$, as specified in the previous definition.
}\end{Remark}

In the rest of this section, we consider the case of an exponential utility function. With our notation, if the investor buys the contingent claim at the price $p$ and invests the rest of his money in the risk-free asset and in the risky asset, the value function is equal to
\begin{equation*}
V(x-p,\xi)= \exp(-\gamma(x-p))V(0,\xi).
\end{equation*}
If he invests all his money in the risk-free asset and in the risky asset, the value function is equal to
\begin{equation*}
V(x,0)= \exp(-\gamma x)V(0,0).
\end{equation*}
Hence, the Hodges price for the contingent claim $\xi$ is given by the formula
 \begin{equation*}
 p=\frac{1}{\gamma}\ln\Big(\frac{V(0,0)}{V(0,\xi) }\Big) =\frac{1}{\gamma}\ln\Big(\frac{J(0,0)}{J(0,\xi) }\Big).
 \end{equation*}
 since $J(0,\xi)= -V(0,\xi)$.

In the case of Section 3, that is where the strategies take their values in a compact set $C$, we have: 
\begin{Proposition} \label{trois}(Compact case) Suppose that the coefficients are bounded.
Let $(J^\xi_t)$ be the solution of Lipschitz BSDE (\ref{Jk edsr}) and $(J^0_t)$  be the solution of Lipschitz BSDE (\ref{Jk edsr}) with $\xi=0$.  The Hodges price for the contingent claim $\xi$ is given by the formula
 \begin{equation}\label{id}
 p=\frac{1}{\gamma}\ln\Big(\frac{J_0^0}{J_0^\xi}\Big).
 \end{equation}
\end{Proposition}

 \begin{Remark}{\rm
Consequently,  the indifference price is simply given in terms of two Lipschitz BSDEs. This leads to possible numerical computations by applying the results of Bouchard and Elie \cite{boueli08}.\\
Note that in the case where the coefficients are not supposed to be bounded but only satisfy some exponential integrability conditions (see Section 7), Proposition \ref{trois} still holds except that BSDE (\ref{Jk edsr}) is no longer Lipschitz (but still admits a unique solution).

}\end{Remark}

In the non constrained case, without any assumptions on the coefficients, we have
\begin{Proposition}(Non constrained case) 
Let $(J_t^\xi)$ (resp. $(J_t^0)$) be the maximal subsolution of BSDE (\ref{EDSR J}) (resp. with $\xi=0$). The Hodges price for the contingent claim $\xi$ associated with $\Ac$ is given by formula (\ref{id}).
\end{Proposition}
Note that if the coefficient $\beta$ is bounded (but not necessarily the others), the indifference price associated with the set $\Theta_2$ of Delbaen \emph{et al.} \cite{deletal02} and that associated with the set $\Ac$ coincide because the value functions $V(x, 0)$ and $V(x-p, \xi)$ are the same for $\Theta_2$ or $\Ac$. 

Recall that in the case of bounded coefficients, $(J_t^\xi)$ is the maximal solution of BSDE (\ref{EDSR J}). Also, in this case, we have:
\begin{Proposition}\label{app}(Approximation of the indifference price) Suppose that the coefficients are bounded.
 The Hodges price $p$ for the contingent claim $\xi$ associated with $\Theta_2$ (or equivalently with $\Ac$) satisfies
 \begin{equation*}
 p=\lim\limits_{k\rightarrow \infty}p^k,
 \end{equation*}
 where for each $k$, $p^k$ is the Hodges price associated with the simple set $\Bc^k$ of all strategies bounded by $k$. 
 For each $k$, $p^k$ is given by 
 \begin{equation*}
 p^k=\frac{1}{\gamma}\ln\Big(\frac{J^{k,0}_0}{J^{k,\xi}_0}\Big),
 \end{equation*}
 where $(J^{k,\xi}_t)$ (resp. $(J^{k,0}_t)$) is the solution of Lipschitz BSDE (\ref{Jk edsr}) (resp. with $\xi=0$) with $\Cc=\Bc^k$. 
 \end{Proposition}
   
 \begin{Remark}{\rm
 This leads to possible numerical computations in order to approximate the indifference price.
Also, note that in the case where the coefficients are not supposed to be bounded but only satisfy some exponential integrability conditions (see Section 7), Proposition \ref{app} still holds except that BSDE (\ref{Jk edsr}) is no longer Lipschitz (but still admits a unique solution).

 }
 \end{Remark}



\section{Generalizations}
\setcounter{equation}{0} \setcounter{Assumption}{0}
\setcounter{Theorem}{0} \setcounter{Proposition}{0}
\setcounter{Corollary}{0} \setcounter{Lemma}{0}
\setcounter{Definition}{0} \setcounter{Remark}{0}

In this section, we give some generalizations of the previous results. The proofs are not given, but they are identical to the proofs of the case with a default time and a stock. In all this section, elements of $\R^n$, $n\geq 1$, are identified to column vectors, the superscript $'$ stands for the transposition, $||.||$ the square norm, $\1$ the vector of $\R^n$ such that each component of this vector is equal to $1$. Let $U$ and $V$ two vectors of $\R^n$, $U*V$ denotes the vector such that $(U*V)_i=U_iV_i$ for each $i\in\{1,\dots,n\}$. Let $X\in \R^n$, $diag(X)$ is the matrix such that $diag(X)_{ij}=X_i$ if $i= j$ else $diag(X)_{ij}=0$.

\subsection{Several default times and several stocks}     \label{plusieur defaut}                
We consider a market defined on the complete probability space $(\Omega,\mathcal{G},\P)$ equipped with two stochastic processes: an $n$-dimensional Brownian motion $(W_t)$ and a $p$-dimensional jump process $(N_t)=((N^i_t),1\leq i \leq p)$ with $N^i_t=\mathds{1}_{\tau^i\leq t}$, where $(\tau^i)_{1\leq i\leq p}$ are $p$ default times. We denote by $\G=\{\Gc_t,0\leq t \leq T\}$ the completed filtration generated by these processes. 
\begin{Assumption}{\rm
We make the following assumptions on the default times:
\begin{enumerate}[(i)]
\item The defaults do not appear simultaneously: $\P(\tau^i=\tau^j)=0$ for $i\neq j$. 
\item Each default can appear at any time: $\P(\tau^i>t)>0$.
\end{enumerate}
}\end{Assumption}
We consider a financial market which consists of one risk-free asset, whose price process is assumed for simplicity to be equal to $1$ at any time, and $n$ risky assets, whose price processes $(S^i_t)_{1\leq i \leq n}$ admit $p$ discontinuities at times $(\tau^j)_{1\leq j \leq p}$. In the sequel, we consider that the price processes $(S^i_t)_{1\leq i \leq n}$ evolve according to the equation
\begin{equation}
dS_t=diag(S_{t^-})(\mu_tdt+\sigma_tdW_t+\beta_tdN_t),  
\end{equation}
with the classical assumptions:
\begin{Assumption}\hfill
{\rm \begin{enumerate}[(i)]
\item $(\mu_t)$, $(\sigma_t)$ and $(\beta_t)$ are $\G$-predictable processes such that $\sigma_t$ is nonsingular for any $t\in[0,T]$ and
\begin{equation*}
\int_0^T||\sigma_t||^2dt+\sum_{i,j}\int_0^T\lambda_t^j|\beta^{i,j}_t|^2dt<\infty~a.s.,
\end{equation*}
\item there exist $d$ coefficients $\theta^1,\dots, \theta^d$ that are $\G$-predictable processes such that 
\begin{equation*}
\mu^i_t+\sum_{j=1}^p\lambda_t^j\beta_t^{i,j}=\sum_{j=1}^d\sigma_t^{i,j}\theta_t^j,~\forall\, t\in[0,T]~a.s.,~1\leq i \leq n;
\end{equation*}
we suppose that $\theta^j$ is bounded,
\item the processes $(\beta^{i,j}_t)$ satisfy $\beta^{i,j}_{\tau_j}>-1$ a.s., for each $i$ and $j$.
\end{enumerate}
}\end{Assumption}

Using the same technics as in the previous sections, we can generalize all the results stated in the previous sections to this framework. In particular, in the classical case of bounded coefficients, if $(J_t)$ denotes the dynamic value function associated with the admissible sets $\Ac$ or $\Ac'$ which are equal, we have: 
\begin{Theorem}
There exist $Z\in L^2(W)$ and $U\in L^2(M)$ such that $(J_t,Z_t,U_t)$ is the maximal solution in $\Sc^{+,\infty} \times L^2(W)\times L^2(M)$ of the BSDE
\begin{equation*}
\left\{\begin{aligned} 
-\,dJ_t=&~\essinf\limits_{\pi\in \mathcal{A}}\Big\{\frac{\gamma^2}{2}||\pi'_t\sigma_t||^2J_t- \gamma\pi'_t(\mu_tJ_t+\sigma_tZ_t)-(\1-e^{-\gamma\pi'_t\beta_t})(\lambda_tJ_t+\lambda_t*U_t)\Big\}dt\\
&~-\, Z_tdW_t-U_tdM_t,   \\
 J_T=&~\exp(-\gamma \xi).
\end{aligned}\right.
\end{equation*}
\end{Theorem}

\begin{Remark}{\rm
The value function $J_0$ coincides with the value function associated with the set $\Theta_2$. 
}\end{Remark}

\subsection{Poisson jumps}
We consider a market defined on the complete probability space $(\Omega,\Gc,\P)$ equipped with two independent processes: a unidimensional Brownian motion $(W_t)$ and a real-valued Poisson point process $p$ defined on $[0,T]\times\R \backslash \{0\}$, we denote by $N_p(ds,dx)$ the associated counting measure, such that its compensator is $\hat{N}_p(ds,dx)=n(dx)ds$ and the Levy measure $n(dx)$ is positive and satisfies $n(\{0\})=0$ and $\int_{\R\backslash\{0\}}(1\wedge |x|)^2n(dx)<\infty$. We denote by $\G=\{\Gc_t,0\leq t \leq T\}$ the completed filtration generated by the two processes $(W_t)$ and $(N_p)$. We denote by $\tilde{N}_p(ds,dx)$ ($\tilde{N}_p(ds,dx)=N_p(ds,dx)-\hat{N}_p(ds,dx)$) the compensated measure, which is a martingale random measure. In particular, for any predictable and locally square integrable process $(U_t)$, the stochastic integral $\int U_s(x)\tilde{N}_p(ds,dx)$ is a locally square integrable martingale. Let us introduce the classical set $L^2(\tilde{N}_p)$ (resp. $L^2_{loc}(\tilde{N}_p)$) given by the set of $\G$-predictable processes on $[0,T]$ under $\P$ with
\begin{equation*}
\E\Big[\int_0^T\int_{\R\backslash\{0\}}|U_t(x)|^2n(dx)dt\Big]<\infty~\text{(resp. }\int_0^T\int_{\R\backslash\{0\}}|U_t(x)|^2n(dx)dt<\infty~a.s.\text{).}
\end{equation*}
The financial market consists of one risk-free asset, whose price process is assumed to be equal to $1$, and one single risky asset, whose price process is denoted by $S$. In particular, the stock price process satisfies
\begin{equation*}
dS_t=S_{t^-}\Big(\mu_tdt+\sigma_tdW_t+\int_{\R\backslash\{0\}}\beta_t(x)N_p(dt,dx)\Big).                   
\end{equation*}
All processes $(\mu_t)$, $(\sigma_t)$ and $(\beta_t)$ are assumed to be $\G$-predictable, the process $(\sigma_t)$ satisfies $\sigma_t>0$ and the process $(\beta_t)$ satisfies $\beta_t(x)>-1$ a.s. Moreover we suppose that
\begin{equation*}
\int_0^T|\sigma_t|^2dt+\int_0^t\int_{\R\backslash\{0\}}|\beta_t(x)|^2n(dx)dt+\int_0^T\Big|\frac{\mu_t+\int_{\R\backslash\{0\}}\beta_t(x)n(dx)ds}{\sigma_t}\Big|^2dt<\infty~a.s.
\end{equation*}

Using the same technics as in the previous sections, we can generalize all the results stated in the previous sections to this framework. In particular, in the classical case of bounded coefficients, if $(J_t)$ denotes the dynamic value function associated with the admissible sets $\Ac$ or $\Ac'$ which are equal, we have:
\begin{Theorem}
There exist $Z\in L^2(W)$ and $U\in L^2(\tilde{N}_p)$ such that $(J_t,Z_t,U_t)$ is the maximal solution in $\Sc^{+,\infty} \times L^2(W)\times L^2(\tilde{N}_p)$ of the BSDE
\begin{equation*}
\left\{\begin{aligned}
-\,dJ_t=&~\essinf\limits_{\pi\in \mathcal{A}}\Big\{\frac{\gamma^2}{2}|\pi_t\sigma_t|^2J_t- \gamma\pi_t(\mu_tJ_t+\sigma_tZ_t)-\int_{\R\backslash\{0\}}(1-e^{-\gamma \pi_tx})(J_t+U_t(x))n(dx)\Big\}dt \\
&~-\,Z_tdW_t-\int_{\R\backslash\{0\}}U_t(x)\tilde{N}_p(dt,dx) , \\
 J_T=&~\exp(-\gamma \xi),
\end{aligned}\right.
\end{equation*}
\end{Theorem}

\begin{Remark}{\rm
The value function $J_0$ coincides with the value function associated with the set $\Theta_2$. 
}\end{Remark}

\vspace{1cm}

\appendix
\begin{center}{\Large{\textbf{Appendix}}}\end{center}

\section{Essential supremum}\label{Neveu}
\setcounter{equation}{0} \setcounter{Assumption}{0}
\setcounter{Theorem}{0} \setcounter{Proposition}{0}
\setcounter{Corollary}{0} \setcounter{Lemma}{0}
\setcounter{Definition}{0} \setcounter{Remark}{0}

Recall the following classical result (see Neveu \cite{nev75}):
\begin{Theorem}
Let $F$ be a non empty family of measurable real valued functions $f:\Omega\rightarrow \bar{\R}$ defined on a probability space $(\Omega,\Fc,\P)$. Then there exists a measurable function $g:\Omega\rightarrow \bar{\R}$ such that 
\begin{enumerate}[(i)]
\item for all $f\in F,~f\leq g$ a.s.,
\item if $h$ is a measurable function satisfying $f\leq h$ a.s., for all $f\in F$, then $g\leq h$ a.s. 
\end{enumerate}
This function $g$, which is unique a.s., is called the essential supremum of $F$ and is denoted $\esssup_{f \in F} f$.\\
Moreover there exists at least one sequence $(f_n)$ in $F$ such that $\esssup_{f\in F}f=\lim_{n\rightarrow \infty} f_n$ a.s. Furthermore, if $F$ is \emph{filtrante croissante} (i.e. $f,\,g\in F$ then there exists $h\in F$ such that both $f\leq h$ a.s., and $g\leq h$ a.s.), then the sequence $(f_n)$ may be taken nondecreasing and $\esssup_{f\in F} f=\lim_{n\rightarrow \infty}\uparrow f_n$ a.s.
\end{Theorem}

\section{A classical lemma of analysis}
\label{appendice fonction affine}
\setcounter{equation}{0} \setcounter{Assumption}{0}
\setcounter{Theorem}{0} \setcounter{Proposition}{0}
\setcounter{Corollary}{0} \setcounter{Lemma}{0}
\setcounter{Definition}{0} \setcounter{Remark}{0}

\begin{Lemma}
The supremum of affine functions, whose coefficients are bounded by a constant $c>0$, is Lipschitz and the Lipschitz constant is equal to $c$.\\
More precisely, let $\Ac$ be the set of $[-c,c]^{n}\times[-k,k]$. Then, the function $f$ defined for any $y\in \R^n$ by 
\begin{equation*}
f(y)=\sup\limits_{(a,b)\in\Ac} \{a.y+b\}
\end{equation*}
is Lipschitz with Lipschitz constant $c$.
\label{fonction affine}
\end{Lemma}

\begin{proof}
\begin{equation*}
\sup\limits_{(a,b)\in\Ac} \{a.y+b\}\leq \sup\limits_{(a,b)\in\Ac} \{a.(y-y')\}+\sup\limits_{(a,b)\in\Ac} \{a.y'+b\}.
\end{equation*}
Which implies 
\begin{equation*}
f(y)-f(y')\leq c ||y-y'||.
\end{equation*}
By symmetry, we have also
\begin{equation*}
f(y')-f(y)\leq c ||y-y'||,
\end{equation*}
which gives the desired result.
\end{proof}

\section{Proof of the closedness by binding of $\Ac'$} \label{Theta}
\setcounter{equation}{0} \setcounter{Assumption}{0}
\setcounter{Theorem}{0} \setcounter{Proposition}{0}
\setcounter{Corollary}{0} \setcounter{Lemma}{0}
\setcounter{Definition}{0} \setcounter{Remark}{0}

\begin{Lemma}
Let $\pi^1, \pi^2$ be two admissible strategies of $\Ac^{'}$ and $s\in [0,T]$. The strategy $\pi^3$ defined by
\begin{equation*}
\pi^3_t=
\begin{cases}
\pi^1_t& \text{if } t\leq s,\\
\pi^2_t& \text{if } t> s,
\end{cases}
\end{equation*}
belongs to $\Ac^{'}$.
\end{Lemma}

\begin{proof}
For any $u\in[0,T]$, we have for any $p>1$
\begin{enumerate}[(i)]
\item if $u>s$, then 
\begin{equation*}
\E[\sup_{r\in[u,T]} \exp(-\gamma pX_r^{u,\pi^3})]=\E[\sup_{r\in[u,T]} \exp(-\gamma pX_r^{u,\pi^2})] <\infty,
\end{equation*}
\item if $u \leq s$, then 
\begin{multline*}
\E[\sup_{r\in[u,T]} \exp(-\gamma pX_r^{u,\pi^3})]  \leq \E[\sup_{r\in[u,T]} \exp(-\gamma pX_r^{u,\pi^1})]\\+\E[\sup_{r\in[s,T]} \exp(-\gamma p(X_s^{u,\pi^1}+X_r^{s,\pi^2}))].
\end{multline*}
By Cauchy-Schwarz inequality,
\begin{multline*}
\E[\sup_{r\in[s,T]} \exp(-\gamma p(X_s^{u,\pi^1}+X_r^{s,\pi^2}))]\leq \E[\sup_{r\in[u,T]} \exp(-2\gamma pX_r^{u,\pi^1})]^{1/2}\\ \times \E[\sup_{r\in[s,T]} \exp(-2\gamma pX_r^{s,\pi^2})]^{1/2}.
\end{multline*}
Hence, $\E[\sup_{r\in[u,T]} \exp(-\gamma pX_r^{u,\pi^3})]<\infty$.
\end{enumerate}
\end{proof}

\section{Proof of the existence of a càd-làg modification of $(J_t)$}
\label{preuve de cadlag}
\setcounter{equation}{0} \setcounter{Assumption}{0}
\setcounter{Theorem}{0} \setcounter{Proposition}{0}
\setcounter{Corollary}{0} \setcounter{Lemma}{0}
\setcounter{Definition}{0} \setcounter{Remark}{0}

The proof is not so simple since we do not know if there exists an optimal strategy in $\Ac$.
Let $\mathbb{D}=[0,T] \cap \mathbb{Q}$, where $ \mathbb{Q}$ is the set of rational numbers. Since $(J(t))$ is a submartingale, the mapping $t \rightarrow J(t,\omega)$ defined on 
$\mathbb{D}$ has for almost every $\omega \in \Omega$ and for any $t$ of $[0,T[$ a finite right limit
\begin{equation*}
J(t^+,\omega)=\lim\limits_{s\in\mathbb{D}, s\downarrow t} J(s,\omega),
\end{equation*}
(see Karatzas and Shreve \cite{karshr91}, Proposition 1.3.14 or Dellacherie and Meyer \cite{delmey80}, Chapter 6). Note that it is possible to define $J(t^+,\omega)$ for any $(t,\omega) \in [0,T]\times \Omega$ by $J(T^+,\omega):= J(T,\omega)$ and
\begin{equation*}
J(t^+,\omega):= \limsup\limits_{s\in\mathbb{D}, s\downarrow t} J(s,\omega),~ t\in[0,T[.
\end{equation*}
From the right-continuity of the filtration $(\Gc_t)$, the process $(J(t^+))$ is $\G$-adapted. It is possible to show that $(J(t^+))$ is a $\G$-submartingale and even that the process $(\exp(-\gamma X_t^\pi)J(t^+))$ is a $\G$-submartingale for any $\pi\in \Ac$. Indeed, from Proposition \ref{le plus grand}, for any $s\leq t$ and for each sequence of rational numbers $(t_n)_{n\geq 1}$ converging down to $t$, we have
\begin{equation*}
\mathbb{E}\big[\exp(-\gamma X_{t_n}^\pi)J(t_n)\big|\Gc_s\big]\geq \exp(-\gamma X_{s}^\pi)J(s)~a.s.
\end{equation*}
Let $n$ tend to $+\infty$. By the Lebesgue theorem, we have that for any $s \leq t$,
\begin{equation}\label{AL}
\mathbb{E}\big[\exp(-\gamma X_{t}^\pi)J(t^+)\big|\Gc_s\big]\geq \exp(-\gamma X_{s}^\pi)J(s)~a.s.
\end{equation}
This clearly implies that for any $s \leq t$, $\mathbb{E}[\exp(-\gamma X_{t}^\pi)J(t^+)|\Gc_s]\geq \exp(-\gamma X_{s}^\pi)J(s^+)$ a.s.,
which gives the submartingale property of the process $(\exp(-\gamma X_t^\pi)J(t^+))$.
Using the right-continuity of the filtration $(\Gc_t)$ and inequality (\ref{AL}) applied to $\pi = 0$ and $s=t$, we get
\begin{equation*}
J(t^+)=\mathbb{E}\big[J(t^+)\big|\Gc_t\big]\geq J(t)~a.s.
\end{equation*}
On the other hand, by the characterization of $(J(t))$ (see Proposition \ref{le plus grand}), and since the process $(\exp(-\gamma X_t^\pi)J(t^+))$ is a $\G$-submartingale for any $\pi\in \Ac$, we have that for any $t\in[0,T]$,
\begin{equation*}
J(t^+)\leq J(t)~a.s.
\end{equation*}
Thus, for any $t\in[0,T]$,
\begin{equation*}
J(t^+)=J(t)~a.s.
\end{equation*}
Furthermore, the process $(J(t^+))$ is c\`{a}d-l\`{a}g. The result follows by taking $J_t=J(t^+)$.

\section{Proof of equality (\ref{esssupA=esssupAbar})}\label{esssup A Abar}
\setcounter{equation}{0} \setcounter{Assumption}{0}
\setcounter{Theorem}{0} \setcounter{Proposition}{0}
\setcounter{Corollary}{0} \setcounter{Lemma}{0}
\setcounter{Definition}{0} \setcounter{Remark}{0}

For any $\pi\in\Ac$, we define the strategy $\pi^k_t=\pi_t\mathds{1}_{|\pi_t|\leq k}$ for each $k\in\N$. The strategy $\pi^k$ is uniformly bounded but not necessarily admissible.                               
For that we define for each $(k,n)\in\N\times \N$ the stopping time
\begin{equation*}
\tau_{k,n}:=\inf\{t,|X_t^{\pi^k}|\geq n\}
\end{equation*}
and the strategy $\pi^{k,n}_t :=\pi^k_t\mathds{1}_{t \leq \tau_{k,n}}$. By construction, it is clear that the strategy $\pi^{k,n}\in\Ac^k$ for each $(k,n)$. Since $\pi_t=\lim_{k}\lim_{n}\pi^{k,n}_t$ $dt \otimes d\P$ a.s., the following equality
\begin{multline*}
\essinf\limits_{\pi\in \bar{\mathcal{A}}}\Big\{\frac{\gamma^2}{2}\pi_t^2\sigma_t^2\bar{J}_t-\gamma\pi_t(\mu_t\bar{J}_t+\sigma_t\bar{Z}_t)-\lambda_t(1-e^{-\gamma\pi_t\beta_t})(\bar{J}_t+\bar{U}_t)\Big\}=\\
\essinf\limits_{\pi\in \mathcal{A}}\Big\{\frac{\gamma^2}{2}\pi_t^2\sigma_t^2\bar{J}_t\gamma\pi_t-(\mu_t\bar{J}_t+\sigma_t\bar{Z}_t)-\lambda_t(1-e^{-\gamma\pi_t\beta_t})(\bar{J}_t+\bar{U}_t) \Big\}
\end{multline*}
holds $dt \otimes d\P$ a.s.

\section{Proof of optimality criterion (Proposition \ref{J martingale})}
\label{Proof of Proposition J martingale}
\setcounter{equation}{0} \setcounter{Assumption}{0}
\setcounter{Theorem}{0} \setcounter{Proposition}{0}
\setcounter{Corollary}{0} \setcounter{Lemma}{0}
\setcounter{Definition}{0} \setcounter{Remark}{0}

Suppose (\emph{i}). Hence, 
\begin{equation*}
J(0)=\inf\limits_{\pi\in\Ac}\E\big[\exp\big(-\gamma (X_T^\pi+\xi)\big)\big]=\E\big[\exp\big(-\gamma (X_T^{\hat{\pi}}+\xi)\big)\big].
\end{equation*}
As the process $(\exp(-\gamma X_t^{\hat{\pi}})J(t))$ is a submartingale and as $J(0)=\E[\exp(-\gamma (X_T^{\hat{\pi}}+\xi))]$, it follows that $(\exp(-\gamma X_t^{\hat{\pi}})J(t))$ is a martingale. \\
To show the converse, suppose that the process $(\exp(-\gamma X_t^{\hat{\pi}})J(t))$ is a martingale. Then, $\E[\exp(-\gamma X_T^{\hat{\pi}})J(T)]=J(0)$. Also, since the process $(\exp(-\gamma X_t^{\pi})J(t))$ is a submartingale for any $\pi\in \Ac$ and since $J(T)=\exp(-\gamma \xi)$, it is clear that $J(0)\leq \inf\limits_{\pi\in\Ac}\E[\exp(-\gamma (X_T^\pi+\xi))]$. Consequently,
\begin{equation*}
J(0)=\inf\limits_{\pi\in\Ac}\E\big[\exp\big(-\gamma (X_T^\pi+\xi)\big)\big]=\E\big[\exp\big(-\gamma (X_T^{\hat{\pi}}+\xi)\big)\big],
\end{equation*}
thus $\hat{\pi}$ is an optimal strategy.

\section{Characterization of the value function as the maximum solution of BSDE (\ref{Jk edsr})}\label{appendice edsr}
\setcounter{equation}{0} \setcounter{Assumption}{0}
\setcounter{Theorem}{0} \setcounter{Proposition}{0}
\setcounter{Corollary}{0} \setcounter{Lemma}{0}
\setcounter{Definition}{0} \setcounter{Remark}{0}

\noindent{\bf Step 1}:
Let us show that there exist $Z \in L^2(W)$ and $U \in L^2(M)$ such that $(J _t,Z _t,U _t)$ is a solution in $\Sc^{+,\infty} \times  L^2(W) \times  L^2(M)$ of BSDE (\ref{Jk edsr}).

From the Doob-Meyer decomposition, since the process $(J _t)$ is a bounded submartingale, there exist $Z \in L^2(W)$, $U \in L^2(M)$ and $(A _t)$ a nondecreasing process with $A _0=0$ such that
\begin{equation*}
dJ _t=Z _tdW_t+U _tdM_t+dA _t.
\end{equation*}
By the same technics as in the proof of Proposition \ref{A egal sup}, since for any $\pi\in \Cc $ the process $(\exp(-\gamma X^\pi_t)J (t))$ is a submartingale, we have
\begin{equation*}
dA _t\geq \esssup\limits_{\pi\in \Cc }\Big\{ \gamma\pi_t(\mu_tJ _t+\sigma_tZ _t)+\lambda_t(1-e^{-\gamma\pi_t\beta_t})(J _t+U _t)-\frac{\gamma^2}{2}\pi_t^2\sigma_t^2J _t\Big\}dt.
\end{equation*}
Since there exists an optimal strategy $\hat{\pi}\in \Cc $ from Proposition \ref{0-optimal}, the optimality criterion gives 
\begin{equation*}
dA _t=\Big\{ \gamma\hat{\pi}_t(\mu_tJ _t+\sigma_tZ _t)+\lambda_t(1-e^{-\gamma\hat{\pi}_t\beta_t})(J _t+U _t)-\frac{\gamma^2}{2}\hat{\pi}_t^2\sigma_t^2J _t\Big\}dt,
\end{equation*}
which implies
\begin{equation*}
dA _t= \esssup\limits_{\pi\in \Cc }\Big\{ \gamma\pi_t(\mu_tJ _t+\sigma_tZ _t)+\lambda_t(1-e^{-\gamma\pi_t\beta_t})(J _t+U _t)-\frac{\gamma^2}{2}\pi_t^2\sigma_t^2J _t\Big\}dt,
\end{equation*}
and $(J _t,Z _t,U _t)$ is solution of BSDE (\ref{Jk edsr}).\\

\noindent{\bf Step 2}:
Using similar arguments as in the proof of Theorem \ref{solution maximale}, one can derive that $(J _t,Z _t,U _t)$ is the \emph{maximal solution} in $\Sc^{+,\infty} \times L^2(W)\times L^2(M)$ of BSDE (\ref{Jk edsr}).

\end{document}